\documentclass[12pt]{article}

\usepackage{cite}
\usepackage{fullpage,setspace,hyperref}

\usepackage{amssymb,amstext,amsmath,amsthm,amsfonts,amsmath,amscd,relsize,url,mathtools}
\usepackage{accents}
\usepackage{enumitem}
\usepackage{bbm} 

\usepackage{color}
\usepackage{booktabs}

\allowdisplaybreaks
\numberwithin{equation}{section}

\newcommand{\be}{\begin{equation}}
\newcommand{\ee}{\end{equation}}
\newcommand{\bea}{\begin{eqnarray}}
\newcommand{\eea}{\end{eqnarray}}

\newcommand{\pa}{\partial}

\newcommand{\rd}{\mathrm{d}}

\newcommand{\Db}{\mathcal{D}}		
\newcommand{\DB}{D^\mathrm{B}}	
\newcommand{\DBt}{\tilde{D}^\mathrm{B}} 
\newcommand{\T}{\mathsf{T}}

\newcommand{\id}{\mathbbm{1}}

\newcommand{\n}{\nabla}					
\newcommand{\nc}{\nabla^c}			
\newcommand{\lc}{\mathring{\n}}		

\newcommand{\Lt}{{\tilde{L}}}
\newcommand{\Pt}{{\tilde{P}}}
\newcommand{\Hh}{\mathcal{H}}
\newcommand{\cF}{\mathcal{F}}
\newcommand{\HH}{\mathcal{H}}
\newcommand{\NN}{\mathcal{N}}
\newcommand{\OO}{\mathcal{O}}
\newcommand{\Th}{{\mathcal{T}}}
\newcommand{\PS}{\mathcal{P}}

\newcommand{\Tt}{\mathbb{T}}
\newcommand{\Lie}{\mathcal L}

\newcommand{\XX}{\mathfrak{X}}
\newcommand{\ap}{\alpha}
\newcommand{\bt}{\beta}

\newcommand{\brac}{[\ ,\ ]}
\newcommand{\bracd}{\bl\ ,\ \br}
\newcommand{\se}{\Gamma}
\newcommand{\bl}{[\![}
\newcommand{\br}{]\!]}
\newcommand{\la}{\langle}
\newcommand{\ra}{\rangle}
\def\to{\rightarrow}
\def\tl{\tilde}

\newcommand{\xt}{\tl{x}}
\newcommand{\yt}{\tl{y}}
\newcommand{\zt}{\tl{z}}
\newcommand{\cD}{\mathcal{D}}

\def\cyc{\sum_{\mathclap{(X,Y,Z)}}}

\newcommand{\nB}{\n^{\mathrm{B}}}

\newcommand{\OB}{\Omega^{\mathrm B}}

\newtheorem{Thm}{Theorem}
\newtheorem{Lem}{Lemma}
\newtheorem{Cor}{Corollary}
\newtheorem{Prop}{Proposition}
\newtheorem{Def}{Definition}
\newtheorem*{Ex}{Example}
\newtheorem{Rem}{Remark}

\newcommand{\qandq}{\quad\mathrm{and}\quad}

\DeclareMathOperator{\End}{End}

\begin{document}

\begin{titlepage}
\vfill

\begin{flushright}
LMU-ASC 37/18 \\
\end{flushright}

\vfill

\begin{center}
   \baselineskip=16pt
   	{\Large \bf A Unique Connection for Born Geometry}
   	\vskip 2cm
   	{\sc  Laurent Freidel$^*$\footnote{\tt lfreidel@perimeterinstitute.ca}, 
          Felix J. Rudolph$^{\dagger}$\footnote{\tt felix.rudolph@lmu.de}, 
          David Svoboda$^*$\footnote{\tt dsvoboda@perimeterinstitute.ca} }  
	\vskip .6cm
    {\small  \it 
        $*$ Perimeter Institute for Theoretical Physics, \\
          31 Caroline St. N.,  Waterloo ON, N2L 2Y5, Canada \\ \ \\
        $\dagger$ Arnold Sommerfeld Center for Theoretical Physics, Department f\"ur Physik, \\
          Ludwig-Maximilians-Universit\"at M\"unchen, Theresienstr. 37, 80333 M\"unchen, Germany }
	\vskip 2cm
\end{center}

\begin{abstract}
\noindent It has been known for a while that the effective geometrical description of compactified strings on $d$-dimensional target spaces implies a generalization of geometry with a doubling of the sets of tangent space directions. This generalized geometry involves an $O(d,d)$ pairing $\eta$ and an $O(2d)$ generalized metric $\HH$. 
More recently it has been shown that in order to include T-duality as an effective symmetry, the generalized geometry also needs to carry a phase space structure or more generally a para-Hermitian structure encoded into a skew-symmetric pairing  $\omega$. The consistency of string dynamics requires this geometry to satisfy a set of compatibility relations that form what we call a Born geometry.
In this work we prove an analogue of the fundamental theorem of Riemannian geometry for Born geometry. We show that there exists a unique connection which preserves the Born structure $(\eta,\omega,\HH)$ and which is torsionless in a generalized sense. This resolves a fundamental ambiguity that is present in the double field theory formulation of effective string dynamics.
\end{abstract}

\vfill

\setcounter{footnote}{0}
\end{titlepage}

\setcounter{tocdepth}{2}
\tableofcontents

\section{Introduction}

''Physics is geometry'' as the saying goes, expresses that  there has been a long and fruitful interplay between inventions of   new geometrical structures and the discovery of new  physical concepts. Indeed, quantum mechanics is intimately tied up with symplectic geometry, general relativity with Riemannian geometry, and the gauge principle with the geometry of principal bundles. The deeper reason behind the connection between physics and geometry is that each time a new geometrical concept was realized in mathematics, a new expression of the relativity principle was at play in fundamental physics. Unification of electric and magnetic, unification of wave and particle or unification of space and time always comes with a relativization of what was before understood as an absolute concept. The mathematical expression of such relativization is geometrical by essence and always reveals a new mathematical structure as the central element of the underlying geometry: the metric for gravity, the symplectic potential for quantum mechanics or the gauge group for Yang-Mills theory.

The key questions underlying our work are the following ones: What could be the new geometrical structure that encompasses the key ideas behind quantum gravity? What extension of the relativity principle does it correspond to? And is there a model that can guide us through the maze of new concepts?

\paragraph{Generalized geometry} In the recent years we have learned more about these questions. It is clear that we expect a form of geometrical unification that underlies the unification of matter and geometry. It is also clear that a well-defined mathematical candidate for such a geometrical unification has now emerged under the name of generalized geometry \cite{Dorfman,Courant-Weinstein, courant1990dirac, Loday, Liu-Weinstein, Roytenberg, Alekseev, Severa:2001qm, Hitchin:2004ut, Gualtieri:2003dx, Vaisman:2004msa, Hitchin:2005in, Gualtieri:2007bq, Xu,Severa:2017oew} (we refer the reader to \cite{Jurco:2016emw} for an enlightening review of generalized geometry and more complete references).

The reason behind generalized geometry attracting the attention of both mathematicians and physicists is its power to unify the description of various objects known from the usual differential geometry perspective. As we we will review later, the key new element of generalized geometry is the generalization of the concept of a differentiable structure encoded into the replacement of the usual Lie bracket by the more general Dorfman bracket \cite{Dorfman}. Alternatively this corresponds to the replacement of the concept of a Lie algebra by the more general notion of Leibniz algebra \cite{Loday}.

On the physical side the new relativity principle that is now being explored is called the principle of relative locality \cite{AmelinoCamelia:2011bm, AmelinoCamelia:2011pe,AmelinoCamelia:2008qg,Freidel:2013zga,Barcaroli:2015xda,Freidel:2016pls,Guerin:2018fja} which proposes to explore the new idea that spacetime is not an absolute concept. Instead it depends on the type of physical probes used, in particular their quantum number such as their momenta and energy. The geometrization of this idea is just starting to take place and naturally leads to the construction of spacetime as a Lagrangian subspace of a para-Hermitian manifold \cite{Cruceanu,Etayo:2004lja,Vaisman:2012ke,Freidel:2017yuv, Svoboda:2018rci, Chatzistavrakidis:2018ztm}. Geometrically, this structure requires the compatibility of an almost symplectic structure $\omega$ and an $O(d,d)$ metrical structure $\eta$.

\paragraph{Double field theory} The model that has guided us so far has been the effective description of string theory on compact target spaces. The fundamental question which has been recently explored more deeply is:  What is the geometry of string theory, especially in the regime where T-duality is manifest? At a deeper level this project can be understood as a quest towards the geometrization of the renormalization group flow of 2-dimensional field theories beyond the fundamental result of Friedan \cite{Friedan:1980jf} which obtains Einstein's equation as a fixed point equation for a non-linear $\sigma$-model on non-compact target spaces.

This has led to a set of new interrelated ideas and models. The most developed one being double field theory (DFT)\cite{Siegel:1993xq,Siegel:1993th, Alvarez:2000bh, Alvarez:2000bi, Hull:2004in, Ellwood:2006ya, Hull:2006va, Grana:2008yw, Hull:2009mi,Hohm:2010pp,Coimbra:2011nw,Aldazabal:2011nj,Hohm:2011si,Hohm:2012mf,Aldazabal:2013sca,Berman:2013uda, Cederwall:2014kxa} which postulates that the effective description of strings compactified on torii of radius $R$ is in terms of fields that live in a doubled space. The coordinates of this doubled space, $X^M=(x^\mu,\xt_\mu)$ split into ordinary spacetime coordinates $x^\mu$ which describe the effective IR geometry in the limit $R\to \infty$ and dual coordinates $\xt_\mu$ which describe the effective UV geometry in the limit $R\to 0$. The index structure on these coordinates implies that there is a fundamental duality pairing $\eta$ between space and dual space which has the structure group $O(d,d)$.

This doubled space is also naturally equipped with a so called generalized metric $\HH$ that includes in its fabric the target space metric $g$ and the Kalb-Ramond $B$-field. The structure group of this metric is $O(2d)$ if one considers only spacelike compactifications. One of the key results in this setting is the result by Hohm and Zwiebach \cite{Hohm:2011si, Hohm:2012mf} and Jeon, Lee and Park \cite{Jeon:2010rw, Jeon:2011cn} which states that the string equations of motion involving the dilaton, spacetime metric and Kalb-Ramond field can simply be written as the vanishing of a generalized Einstein tensor. This tensor is constructed in terms of a connection that preserves both $\eta$ and $\HH$.  A striking result that clearly hints towards the conjecture that generalized geometry captures the key elements of string geometry.

There are, however, two caveats to this picture that our work aims to adress. The first one is that in order for double field theory to function properly and admit the proper invariance under the extended symmetry transformations, one needs to choose a solution to the section condition (see e.g. \cite{Berman:2011jh} for an early general discussion of the section condition) that effectively projects the support of the doubled fields onto a Lagrangian subspace. One issue with the section condition is that it is a purely kinematical constraint that is added to the theory without any dynamical handle or geometrical depth. In our previous paper \cite{Freidel:2017yuv} we have shown how this could be resolved by constructing the Dorfman bracket and its projection in terms of the para-Hermitian structure $(\eta,\omega)$. This description also allowed us to give a clear connection between the framework of double field theory that aims to  obtained spacetime as a projection of the doubled space and the mathematical setting of generalized geometry which assumes a pre-existing spacetime with an extended tangent space structure. Once again, the para-Hermitian geometrical structure has proven essential to connect these two different settings.

The second caveat which is the original motivation behind our current project is the fact that although one can reconstruct the string equations as the vanishing of a Ricci tensor, there is no analogue of the fundamental theorem of Riemannian geometry. The fundamental theorem of Riemannian geometry states that on any pseudo-Riemannian manifold there is a unique torsion-free connection -- called the Levi-Civita connection -- which preserves the given metric. Such a theorem fails to be true in generalized geometry since there are many connections that are torsionless in a generalized sense while still preserving $\eta$ and $\HH$. Remarkably, this ambiguity is hidden from us at the one-loop level since it does not affect the construction of the Ricci tensor. The ambiguity disappears in the trace operation \cite{Hohm:2011si, Coimbra:2011ky, Jurco:2016emw}  if we demand also compatibility with the dilaton. While quite remarkable, this is still  deeply unsettling. One would expect that a proper understanding of the geometrical structure uniquely defines the geometry of its connection. Moreover, this is necessary if one wants to go beyond one loop as the Riemann tensor will necessarily be involved in higher order corrections to the string action.

\paragraph{Born geometry} The resolution of this puzzle comes from the simple idea that the reason one cannot fix the connection uniquely is because one geometrical structure in the characterization of the string geometry is missing. Moreover, the structure we are missing is the one needed in order to implement the principle of relativity. This implies that in addition to the metrics $(\eta,\HH)$ we need to include a para-Hermitian structure $(\eta,\omega)$ in the extended geometry which allows us to geometrically define spacetime as one of its Lagrangian subspaces (the one  in the kernel of $\eta-\omega$).

Remarkably, the same conclusion can be reached from the formulation of duality symmetric string theory. The duality symmetric formulation of string theory was first formulated independently by Duff and Tseytlin \cite{Duff:1989tf,Tseytlin:1990nb,Tseytlin:1990va} and further developed in many directions (see \cite{Szabo:2018hhh,Chatzistavrakidis:2018ztm,Berman:2007xn,Sfetsos:2009vt} for a sample and references therein). It is the source of the double field theory developments described earlier. The series of work \cite{Freidel:2013zga, Freidel:2014qna,Freidel:2015pka,Freidel:2015uug} on metastring theory developed further the formulation of string theory compatible with T-duality and deepened our geometrical understanding of string geometry. It was shown that the geometry of the T-dual string exactly requires a triple structure $(\eta,\omega,\HH)$ satisfying appropriate compatibility conditions. This new geometrical structure was called {\it Born geometry}. The new player in Born geometry compared to the geometry of double field theory is the presence of a two-form $\omega$ which promotes the double space to a phase space.

This was recently confirmed in a direct study of compact bosons which showed that indeed the coordinates $x$ and their duals $\tilde{x}$ form a canonical pair that becomes non-commuting at the quantum level \cite{Freidel:2017wst,Freidel:2017nhg}. At the semi-classical level the presence of this extra pre-symplectic structure allows us to define a para-Hermitian structure that structurally enters the definition of the Dorfman bracket and the mathematical definition of the section condition \cite{Freidel:new}.

In the current paper we show that Born geometry is exactly what is needed in order to resolve the ambiguity in the construction of the connection. We prove an analogue of the fundamental theorem of Riemannian geometry for Born geometry, showing that there exists a unique connection -- the Born connection -- which is torsionless in a generalized sense and which preserves $\eta$, $\omega$ and $\HH$. We also show that the structure group of Born geometry is the Lorentz group and that therefore the Born connection reduces to the Levi-Civita connection on its Lagrangian subspace.

The paper is organized as follows: Section 2 provides a concise summary of para-Hermitian geometry, in this section we also define the D-bracket, generalized torsion, the canonical connection and explain the relationship between para-Hermitian connections and D-brackets. In Section 3 we define the notion of Born geometry, introduce the three fundamental structures $(\eta,\omega,\HH)$ and describe their relations. We also prove that the structure group of Born Geometry is simply the Lorentz group. Section 4 forms the main body of this paper where we formulate and prove the main theorem giving the existence and uniqueness of the Born connection. In Section 5 we start exploring the relation between the Born connection and the connections used in DFT and generalized geometry. We prove that the Born connection naturally gives a generalized connection and show that it gives the Levi-Civita connection when projected onto its Lagrangian subspace. Lastly, the appendix contains the proofs of some of the statements in the main text.

\section{Summary of para-Hermitian Geometry}
\label{sec:paraHG}

Let us begin by recalling basic notions of para-Hermitian geometry, one of the building blocks of the Born geometry discussed in Section \ref{sec:borngeo}. A para-Hermitian manifold can be thought of as the real analogue of a Hermitian manifold. The backbone of a para-Hermitian manifold is a para-complex manifold:
\begin{Def}
An almost para-complex manifold is a pair $(\PS,K)$, where $\PS$ is an $2d$-dimen-sional differential manifold and $K\in \End\, T\PS$ such that $K^2=\id$ and the $\pm 1$-eigenbundles of $K$ have the same rank. We say $K$ is integrable and call $(\PS,K)$ a para-complex manifold if the eigenbundles of $K$ are both integrable distributions.
\end{Def}
\begin{Rem}
In the following, we denote the $+1$ and $-1$-eigenbundles of $K$ by $L$ and $\Lt$, respectively. These eigenbundles come with associated projection operators given by 
\begin{align}\label{eq:K-projections}
P=\frac{1}{2}(\id+K) \quad \mathit{and} \quad \Pt=\frac{1}{2}(\id-K).
\end{align}
\end{Rem}

The integrability of the para-complex structure $K$ can be expressed in terms of the Nijenhuis tensor analogous to the complex case:
\begin{Def}
Let $(\PS,K)$ be an almost para-complex manifold. The Nijenhuis tensor associated to $K$ is given by
\begin{align}\label{eq:nijenhuis}
\begin{aligned}
N_K(X,Y)&\coloneqq\frac{1}{4}\Big( [X,Y]+[KX,KY]-K([KX,Y]+[X,KY])\Big)\\
&=\frac{1}{4}\Big( (\n_{KX}K)Y+(\n_XK)KY-(\n_{KY}K)X-(\n_YK)KX\Big)
\\
&=P[\Pt X,\Pt Y]+\Pt[P X,P Y],
\end{aligned}
\end{align}
where $\n$ is any torsionless connection and $[\ , \ ]$ the Lie bracket on $T\PS$.
\end{Def}

One can see that $K$ is integrable if and only if $N_K$ vanishes. An important difference from almost complex manifolds is that the integrability of $L$  is independent from the integrability of $\Lt$: $L$ can be integrable when $\Lt$ is not and vice versa. This leads to the notion of half-integrability: we say an almost para-complex manifold $(\PS,K)$ is $L$-integrable ($\Lt$-integrable) if $L$ ($\Lt$) is an integrable distribution. The half-integrability of $L$ and $\Lt$ is governed by the {\it projected Nijenhuis tensors}
\begin{align}
\begin{aligned}
N_K^P(X,Y)&\coloneqq N_K(PX,PY)=\Pt[PX,PY],\\
N_K^\Pt(X,Y)&\coloneqq N_K(\Pt X,\Pt Y)=P[\Pt X,\Pt Y].
\end{aligned}
\end{align}

\begin{Rem}
Whenever the integrability of $K$ is not relevant to the discussion, we shall refer to any almost para-complex structure simply as para-complex, while emphasizing (non-)in-tegrability of $K$ using the names almost para-complex and integrable para-complex structure whenever a confusion is possible.
\end{Rem}

Now adding a metric compatible with the para-complex structure yields a para-Hermitian structure on the manifold $\PS$.
\begin{Def}
A para-Hermitian manifold is a triple $(\PS,K,\eta)$, where $(\PS,K)$ is a para-complex manifold and $\eta$ is a pseudo-Riemannian metric such that $K$ is skew-orthogonal with respect to $\eta$, i.e. 
\begin{equation}
\eta(KX,KY)=-\eta(X,Y).
\label{eq:etaKK}
\end{equation}
\end{Def}
The same nomenclature as in the case of para-complex manifolds concerning (half-) integrability applies for para-Hermitian manifolds. Since the eigenbundles of $K$ have the same rank and $K$ is skew-orthogonal, the metric $\eta$ is of split signature $(d,d)$. Furthermore, the tensor $\omega\coloneqq \eta K$ is skew
\begin{align}
\omega(X,Y)=\eta(KX,Y)=-\eta(X,KY)=-\omega(Y,X),
\end{align}
and non-degenerate (because $\eta$ is non-degenerate), therefore $\omega$ is an almost symplectic form, called the {\it fundamental form}. The fundamental form is anti-compatible with $K$
\be
\omega(KX,KY)=-\omega(X,Y).
\ee
It follows that the eigenbundles $L$ and $\Lt$ are isotropic with respect to $\eta$ and Lagrangian with respect to  $\omega$.  When $\omega$ is closed, the pair $(\omega,K)$  defines an almost bi-Lagrangian structure \cite{Etayo:2004lja}.

\begin{Rem} \label{rem:canonicalform}
On a para-Hermitian manifold, there always exists a frame \cite[Prop. 1.4]{bejan1993existence} such that $\eta, \omega$ and $K$ take the block-matrix form
\begin{align}
\bar{\eta}&=
\begin{pmatrix}
0 & \id \\
\id & 0
\end{pmatrix},
&
\bar{\omega}&=
\begin{pmatrix}
0 & -\id \\
\id & 0
\end{pmatrix}, 
&
\bar{K}&=
\begin{pmatrix}
\id & 0 \\
0 & -\id
\end{pmatrix}.
\end{align}
In the following we will repeatedly use that  $\eta$ allows us to identify $L$ with ${\Lt}^*$ and $\Lt $ with $L^*$ as this choice of frame exemplifies.
\end{Rem}

The fundamental form $\omega$ can be used to express the para-Hermitian Nijenhuis tensor. We denote $N_K(X,Y,Z) := \eta(N_K(X,Y),Z)$ and find
\be
\cyc N_K(X,Y,Z)= 
\left[\rd \omega^{(3,0)}-\rd\omega^{(0,3)}\right](X,Y,Z).\label{cycN}
\ee
where $(p,q)$ is the standard bigrading of forms corresponding to the splitting of the tangent bundle $T\PS=L\oplus \Lt$ given by $K$. The proof  given for instance in \cite{Freidel:2017yuv} follows from (\ref{eq:nijenhuis}) and the identities
\be\label{domid}
\n_X\omega(Y,Z)=\eta ((\n_X K) Y,Z),\qquad
\n_X\omega(KY,KZ)=\n_X\omega(Y,Z). 
\ee
where $\n$ is any connection compatible with the metric $\eta$. An obvious consequence -- which will be of interest to us later -- is the fact that if $N_K$ is  skew in all indices then it is an exact $(3,0)+(0,3)$-form.

\begin{Rem}
The data $(\PS,K,\eta)$, $(\PS,\eta,\omega)$ and $(\PS,K,\omega)$ are equivalent and so we may use the different triples interchangeably.
\end{Rem}

\subsection{Canonical Connection}
\label{sec:cancon}
In this section we recall the definition and some properties of the canonical connection of a para-Hermitian manifold $(\PS,\eta,K)$ \cite{Freidel:2017yuv}. It was discovered in \cite{Freidel:2017yuv} that this connection can be used to construct a bracket on $T\PS$ -- the D-bracket -- which is related to the Dorfman bracket defining the key notion of Courant algebroids. Properties of this bracket were further investigated in \cite{Svoboda:2018rci}. We will review some of these properties in the subsequent sections.

\begin{Def}\label{def:cancon}
Let $(\PS,\eta,K)$ be a para-Hermitian manifold and let $\lc$ be the Levi-Civita connection of $\eta$. We define the canonical connection $\nc$ by
\begin{align}
\nc_XY=P\lc_X PY+\Pt \lc_X \Pt Y.
\end{align}
The connection $\nc$ can be expressed using its contorsion tensor as
\begin{align}\label{eq:can-con-contorsion}
\eta(\nc_XY,Z)=\eta(\lc_XY,Z)-\frac{1}{2}\lc_X\omega(Y,KZ).
\end{align}
Equivalently, it also can be expanded in the Koszul form as  
\begin{equation}
\begin{aligned}
2\eta(\nc_XY,Z)
&=  X[\eta(Y,Z)]-\eta(X,[Y,Z] + T^c(Y,Z))\\
& + Y[\eta(X,Z)]-\eta(Y,[X,Z] + T^c(X,Z))\\
& -   Z[\eta(X,Y)]+\eta(Z,[X,Y] + T^c(X,Y)),
\end{aligned}
\end{equation}
where $T^c(X,Y)$  is the torsion of the canonical connection.
\end{Def}
The canonical connection is a para-Hermitian connection, meaning that $\nc \omega=\nc \eta=0$. This in particular means that $\nc$ preserves the bundles $L$ and $\Lt$. We will need the following lemma which justifies the central importance of the canonical connection. Defining $T^c(X,Y,Z):= \eta(T^c(X,Y),Z)$
which evaluates to 
\be
 T^c(X,Y,Z)= \tfrac12\left(\lc_Y\omega(X,KZ)- \lc_X\omega(Y,KZ)\right),
\ee one gets
\begin{Lem} \label{polarised} The torsion of the canonical connection is polarised, i.e. $PT^c(PX,PY)= 0= \Pt T^c(\Pt X,\Pt Y)$, and its polar components are minus the components of the Nijenhuis tensor
\begin{equation}
\begin{aligned}
\Pt T^c(PX,PY) &= -\Pt [PX,PY] = -N_K^P(X,Y), \\
P T^c(\Pt X,\Pt Y) &= - P [\Pt X,\Pt Y] = -N_K^\Pt(X,Y).
\end{aligned}
\end{equation} 
The cyclic permutation of the torsion is related to the cyclic permutation of the Nijenhuis tensor via
\be
\cyc \left[T^c(X,Y,Z) + N_K(X,Y,Z)\right] = 
\left[\rd \omega^{(2,1)}-\rd\omega^{(1,2)}\right](X,Y,Z). 
\ee
\end{Lem}
\begin{proof}
The first statement follows from (\ref{domid}) which implies that $\nabla_X \omega(PY,\Pt Z)=0$. The second statement follows from the evaluation
\bea T^c(PX,PY)&=& \nc_{PX}PY-\nc_{PY}PX -[PX,PY]\cr
&=& P( \lc_{PX}PY-\lc_{PY}PX) -[PX,PY],\cr
&=& - \Pt ([PX,PY])= - N_K(PX,PY).
\eea
The third statement follows directly from this after summing over cyclic permutations. It can also be read off from the expression (\ref{cycN}) of the Nijenhuis tensor and the form of the torsion.
\end{proof}


\subsection{D-Bracket}
\label{sec:Dstructure}

In this section we formally introduce the notion of a D-bracket. The D-bracket is a natural bracket operation defined on the tangent bundle of a para-Hermitian manifold and can be thought of as a generalization of the Lie bracket. The fact that a coordinate-free definition of the D-bracket requires a para-Hermitian structure has been first observed by Vaisman in  \cite{Vaisman:2004msa, Vaisman:2012ke, Vaisman:2012px} and further developed in \cite{Freidel:2017yuv,Svoboda:2018rci} along with a detailed discussion of the relationship with the Dorfman bracket \cite{Dorfman}.

In the mathematical literature the Dorfman bracket is defined abstractly in terms of its defining properties \cite{Severa:2017oew}. In the physics literature \cite{Grana:2008yw, Hull:2009zb}, both the usual D-bracket and its twisted versions are understood as a specific bracket operation and as such are defined by explicit formulas, usually in the form of coordinate expressions. 

The aim here is to define the notion of the D-bracket abstractly along with its specialization, the \emph{canonical} D-bracket which matches the definition of the usual D-bracket in physics. The deeper reason behind this abstract approach will be further exposed in a subsequent publication \cite{Freidel:new}. Here we only introduce the definitions and some of the main properties.

\begin{Def}
\label{def:MetricAlgebroid}
A metric-compatible bracket on a pseudo-Riemannian manifold $(\PS,\eta)$  is a bilinear operation 
$\bracd: \XX(\PS)\times \XX(\PS) \to \XX(\PS)$ on the algebra of vector fields, satisfying the following properties for any $X,Y,Z\in \XX(\PS)$ and $f\in C^\infty(\PS)$:
\begin{align}
X[\eta(Y,Z)] &= \eta(\bl X,Y\br,Z)+\eta(Y,\bl X,Z\br), \label{eq:DbracketComp} \\
\bl X,fY\br &= f\bl X,Y\br+X[f]Y, \label{eq:DbracketLeibniz1} \\
\eta(Y,\bl X,X\br ) &= \eta(\bl Y,X\br,X).
\label{eq:DbracketNorm}
\end{align}
\end{Def}
The triple $(T\PS,\eta,\bracd)$ -- where $\bracd$ is compatible with the metric $\eta$ -- defines a so called \textit{metric algebroid} \cite{Vaisman:2012ke} where the anchor\footnote{Note that unlike for a Lie or Courant algebroid, the anchor map for a metric algebroid is not required to satisfy the morphism compatibility condition between the metric algebroid bracket and the Lie bracket.}  is taken to be the identity map on $T\PS$.

The properties of a metric compatible  bracket can be understood in the following way. The Leibniz property \eqref{eq:DbracketLeibniz1} says that the operation $\mathbb{L}_X:=\bl X,\cdot\br$ is a derivation, usually called the generalized Lie derivative\footnote{We use the bracket notation in this paper to denote the derivative since it renders the formulae simpler to read and since we do not use the C-bracket. This should be kept in mind when comparing with results in \cite{Freidel:2017yuv} where the Lie derivative notation is used.}. The condition \eqref{eq:DbracketComp} means that this derivative is compatible with the para-Hermitian metric. It is important to note here that the bracket is not skew-symmetric. The condition \eqref{eq:DbracketNorm} determines the symmetric component of the bracket. In order to unfold the meaning of this condition, it is customary to associate to any function $f$ a vector field $\mathcal{D}[f]$ which is defined by the condition $\eta(\mathcal{D}[f],X)=X[f]$.  Equivalently this means that $\mathcal{D}[f] : = \eta^{-1}\rd f$. From (\ref{eq:DbracketComp}) and (\ref{eq:DbracketNorm}) we conclude that $Y[\eta(X,X)] = 2 \eta(\bl Y,X\br,X)= 2\eta(Y,\bl X,X\br)$ which implies that
\be 
\label{eq:DbracketNorm1}
\bl X,X\br = \frac{1}{2}\mathcal{D}[\eta(X,X)]. 
 \ee
This is the form usually used. Finally we note that the combination of \eqref{eq:DbracketNorm1} and \eqref{eq:DbracketLeibniz1} can be used to determine the other Leibniz property of the bracket
\be
\bl fX,Y\br = f\bl X,Y\br-Y[f]X+\eta(X,Y)\mathcal{D}[f].\label{eq:DbracketLeibniz2} 
\ee

In the spirit of generalizing usual geometrical notions by replacing the Lie bracket with the metric algebroid bracket, we define the generalized Nijenhuis tensor which measures the failure of a para-Hermitian structure to be integrable in the generalized sense:

\begin{Def}\label{def:gen-nij}
Let $(T\PS,\eta,\bracd)$ be a metric algebroid and $C \in \End\,T\PS$ a charge conjugation operator\footnote{$C$ stands for charge conjugation. This follows from the fact that such operations can be used to change the sign of the metastring action \cite{Freidel:2015pka} into its opposite.} such that
\begin{align}\label{eq:gen-nij-requirements}
C^2=\pm \id,\quad \eta(CX,Y)=- \eta(X,CY).
\end{align}
We define the generalized Nijenhuis tensor ${\cal{N}}_C$ associated with $C$ by
\begin{equation}
{\cal{N}}_C(X,Y) : =C^2\bl X,Y\br +\bl C X,C Y\br  - C\big(\bl C X,Y\br  + \bl X,C Y\br\big).
\label{eq:gen-nij-def}
\end{equation}
\end{Def}

One first has to check that the conditions (\ref{eq:gen-nij-requirements}) ensure that ${\cal{N}}_C$ is skew-symmetric even if the bracket is not. The fact that the generalized Nijenhuis tensor is indeed a tensor follows from then the usual evaluation
\begin{align}
{\cal{N}}_C(X,fY)&=f{\cal{N}}_C(X,Y)\pm X[f]Y +C(X)[f]C(Y)  - C\big(C(X)[f]Y  +  X[f]C(Y)\big) \notag\\
&=f{\cal{N}}_C(X,Y)\pm X[f]Y+C(X)[f]C(Y)-C(X)[f]C(Y)\mp X[f]Y \notag\\
&=f{\cal{N}}_C(X,Y).
\end{align}
Notice that unless the right signs in \eqref{eq:gen-nij-requirements} are chosen, $\mathcal{N}_C$ will not be tensorial. This is because $\bracd$ is not skew-symmetric and therefore the tensoriality in each entry of $\mathcal{N}_C$ give different requirements on the sign choice. This is in contrast with the usual Nijenhuis tensor which does not require $C$ to be appropriately compatible with a metric.

\begin{Def} 
\label{def:Dbracket}
A D-bracket\footnote{D stands for {\it Dorfman} since its definition generalizes the  axioms of the Dorfman bracket. It also stands for  {\it differentiable} since it provides a generalization of the differentiable structure encoded by the Lie bracket to para-Hermitian manifolds. It also could stand for {\it Dirac structure} since both $L$ and $\Lt$ are null and integrable $\bl L, L \br \subset L$ and $\bl \Lt,\Lt\br\subset \Lt$ by definition for a D-bracket. We let the reader decide which interpretation is more appropriate for him or her.} on an almost para-Hermitian manifold $(\PS,\eta, K)$ is a metric-compatible bracket $\bracd$ such 
that $K$ is integrable in the generalized sense, i.e. such that ${\cal{N}}_K=0$. When this is the case we will refer to the data $(\PS,\eta,K,\bracd)$ as a D-structure. 
\end{Def}

The condition to be a D-bracket simply means that $L$ and $\Lt$ are Dirac structures with respect to $\bracd$. This means they are integrable in the generalized sense 
\begin{equation}
\begin{aligned}
\NN_K(PX,PY) &= 4\Pt(\bl PX, PY\br) = 0, \\
\NN_K(\Pt X,\Pt Y) &= 4P(\bl \Pt X, \Pt Y\br) = 0.
\end{aligned}
\label{Generalisedint}
\end{equation}
This generalized integrability property will play a crucial role in the construction of the Born connection. 

Finally, among all possible metric algebroid brackets we single out the ones which are canonical:
\begin{Def}
A D-structure is said to be canonical if it satisfies 
\bea 
\bl PX,PY \br&= P([PX,PY]),\label{eq:comp1}\\
\bl \Pt X,\Pt Y \br&=\Pt([ \Pt X,\Pt Y]) \label{eq:comp2}
\eea
where $[\,,\,]$ denotes the Lie bracket. 
\end{Def}
The conditions \eqref{eq:comp1} and \eqref{eq:comp2} mean that a canonical D-bracket restricted to the Lagrangian $L$ (respectively $\Lt$) is given by the projection of the Lie bracket onto $L$ (respectively $\Lt$). These conditions are compatible with the generalized integrability conditions of $L$ and $\Lt$ \eqref{Generalisedint}.
 
The key role of the D-bracket for the current discussion is that it generalizes the notion of a differentiable structure and allows us to define a natural notion of generalized torsion on $T\PS$, discussed in Section \ref{sec:gentor}.  The existence of a metric compatible bracket and the existence of a canonical D-bracket follows from the following lemma:
\begin{Lem}
There exists a unique canonical D-bracket on $(\PS,\eta,K)$ which is
 given by 
\be\label{eq:D-bracket}
\eta(\bl X,Y\br^c,Z)=\eta(\nc_XY-\nc_YX,Z)+\eta(\nc_ZX,Y)
\ee 
where $\nc$ is the canonical connection (Definition \ref{def:cancon}).
\end{Lem}
The canonical D-bracket and its twisted version, as well as its relationship to the Dorfman bracket and generalized geometry, are discussed in detail in \cite{Freidel:2017yuv, Svoboda:2018rci}. We will expand on its relation to generalized geometry in a forthcoming publication \cite{Freidel:new}. 
\begin{proof}
One starts by proving unicity. Let's assume that $\bracd$ and $\bracd'$ are two canonical D-brackets associated to the same almost para-Hermitian manifold $(\PS,\eta,K)$, denote the difference of D-brackets by 
\be
\bl X, Y\br' = \bl X, Y\br +\Delta\cF(X,Y).
\ee
and define $\Delta\cF(X,Y,Z):= \eta(\Delta\cF(X,Y),Z)$. From the metric compatibility \eqref{eq:DbracketComp} it follows that $\Delta\cF(X,Y,Z)=-\Delta\cF(X,Z,Y)$ is skew in the last two slots. From the symmetry property \eqref{eq:DbracketNorm} it follows that $\Delta\cF(X,Y,Z)=-\Delta\cF(X,Y,Z)$ is skew in the first two entries, hence it is fully skew. From the Leibniz property \eqref{eq:DbracketLeibniz1} it follows that $\Delta\cF(X,Y,fZ)= f\Delta\cF(X,Y,Z)$ is tensorial, hence it is a three-form on $\PS$. From the generalized anchoring property (\ref{eq:comp1}, \ref{eq:comp2}) it follows that $\Delta\cF(PX,PY)=0=\Delta\cF(\Pt X,\Pt Y)$. In other words this means that $\Delta\cF(PX,PY,Z) =0$ and $\Delta\cF(\Pt X,\Pt Y ,Z)=0$ for all $Z \in \XX(\PS)$. Since $\Delta\cF$ is a three-form, this implies that $\Delta\cF=0$ which shows unicity.

The fact that the canonical bracket satisfies the metric compatibility \eqref{eq:DbracketComp}, the symmetry property \eqref{eq:DbracketNorm} and the Leibniz property \eqref{eq:DbracketLeibniz1} was already shown in \cite{Freidel:2017yuv} and follows by direct inspection. The generalized anchoring property follows from 
\bea
 \eta(\bl PX,PY\br^c,Z)&=&\eta(\nc_{PX}PY-\nc_{PY}PX,Z)+\eta(\nc_ZPX,PY)\cr
 &=& \eta(P(\lc_{PX}PY-\lc_{PY}PX),Z)\cr
 &=& \eta(P([{PX},PY]),Z).
\eea
The second equality follows from the definition of the canonical connection, the fact that it preserves $L$ and that $L$ is Lagrangian, while the last equality follows from the torsionlessness of the Levi-Civita connection.
\end{proof}

\subsection{Generalized Torsion}
\label{sec:gentor}

We will now define a notion of generalized torsion \cite{Gualtieri:2007bq,Coimbra:2011nw} which is a tensorial quantity assigned to any connection on a para-Hermitian manifold  equipped with a D-structure $\bracd$. The generalized torsion can be thought of as a direct analogue to the usual torsion tensor of a connection where the Lie bracket $\brac$ is replaced by the canonical D-bracket\footnote{From now on we will only use the canonical D-bracket as a reference bracket, we will therefore drop the superscript $c$. In other words $\bracd\equiv \bracd^c$ for the rest of the paper.} $\bracd$ in an appropriate manner. The usual torsion tensor of a connection $\n$ is defined as
\begin{align}
T_\n(X,Y)=\n_XY-\n_YX-[X,Y],
\end{align}
and measures how much the bracket $[ X,Y]^\n: = \n_XY-\n_YX$ associated with the connection $\n$ differs from the Lie bracket. This leads us to the following generalization:
\begin{Def}\label{def:gen-torsion}
Let $(\PS,\eta,K, \bracd)$ be the canonical D-structure and $\n$ a connection on $T\PS$. The generalized torsion of $\n$ is defined as
\begin{align}
\Th_\n(X,Y,Z):=\eta(\n_XY-\n_YX,Z)+\eta(\n_ZX,Y)-\eta(\bl X,Y\br,Z).
\label{eq:gen-torsion}
\end{align}
\end{Def}
The meaning of the generalized torsion is that it measures how much the bracket $\bracd^\n$ associated to the connection $\n$
\begin{align}
\eta(\bl X,Y\br^\n, Z):=\eta(\n_XY-\n_YX,Z)+\eta(\n_ZX,Y),
\end{align}
deviates from the canonical D-bracket. The generalized torsion then vanishes precisely when the connection $\n$ gives the canonical D-bracket via the above formula. 

A set of key  properties for the generalized torsion which generalizes the skew-symmetry property of the usual torsion is the following:
\begin{Lem}
The generalized torsion $\Th_\n(X,Y,Z)$ of a connection $\n$ is a three-form -- i.e. it is  tensorial and fully skew -- if and only if $\n$ is compatible with $\eta$. When this is the case the bracket $\bracd^\n$ is metric-compatible.

Moreover, if  $\n$ is para-Hermitian -- i.e. it preserves $(\eta, K)$ -- the $(3,0)$ and $(0,3)$ components of $\Th_\n(X,Y,Z)$ vanish and the bracket $\bracd^\n$ is a D-bracket which is not necessarily canonical.
\end{Lem}
\begin{proof}
In order to see this we introduce the  contorsion tensor $\Omega$ associated with $\n$ which measures its deviation from  the canonical connection,
\be
\eta(\n_XY,Z)=\eta(\nc_XY,Z)+\Omega(X,Y,Z).
\ee
The generalized torsion in terms of the contorsion reads
\begin{align}\label{eq:gen-torsion-omega}
\Th^\n(X,Y,Z)=\Omega(X,Y,Z)-\Omega(Y,X,Z)+\Omega(Z,X,Y).
\end{align}
It is easy to see that the contorsion tensor is skew-symmetric in its last last two entries if and only if $\n$ preserves $\eta$,
\begin{align}\label{eq:compatibility-contorsion}
\Omega(X,Y,Z)=-\Omega(X,Z,Y) \iff \n \eta=0.
\end{align}
 Using the skew-symmetry property of $\Omega$, we get 
\begin{align}\label{cyclic}
\Th^\n(X,Y,Z)=\Omega(X,Y,Z)+\Omega(Y,Z,X)+\Omega(Z,X,Y)=\cyc\Omega(X,Y,Z).
\end{align}
The generalized torsion $\Th^\n(X,Y,Z)$ is therefore invariant under cyclic permutations. The fact that is is totally skew follows from skewness of $\Omega$.
For the converse statement, we use that 
\be
\Th(X,Y,Z)+\Th(Y,X,Z)=\Omega(X,Y,Z)+\Omega(X,Z,Y),
\ee
which vanishes if  $\Th$ is fully skew. This yields $\Omega(X,Y,Z)+\Omega(X,Z,Y)=0$ and \eqref{eq:compatibility-contorsion} then shows that $\n$ must be compatible with $\eta$.

If $\n$ is a para-Hermitian connection, i.e. $\n P = P\n$, then
\be
\eta(\n_{PX}PY,PZ) =  \eta(P \n_{PX}PY,PZ) =0.
\ee
From this it is clear that the contorsion satisfies $\Omega(PX,PY,PZ) =0$ and that the $(3,0)$ component of $\Th$ also vanishes. The same argument applies for $\Pt$.  Therefore the generalized torsion of a para-Hermitian connection is a 
$(2,1)+(1,2)$-form. This in turn implies that the bracket $\bracd^\n$ is canonical.
\end{proof}

\begin{Ex}[Generalized torsion of the Levi-Civita connection]
A counter-intuitive fact following from the definition of the generalized torsion is that Levi-Civita now has a torsion (in this generalized sense) unless $\omega$ is closed. This is easy to see from the formula \eqref{eq:can-con-contorsion} where we can explicitly read off the contorsion tensor, $\Omega(X,Y,Z)=\frac{1}{2}\lc_X\omega(Y,KZ)$. Using the formula
\begin{align}
\rd\omega(X,Y,Z)=\cyc \lc_X\omega(Y,Z),
\end{align}
we find  that the generalized torsion is then given by
\begin{align}
\begin{aligned}
\Th_{\lc}(Y,X,Z)=
\frac{1}{2}[&\rd \omega^{(3,0)}+\rd\omega^{(2,1)}-\rd\omega^{(1,2)}-\rd\omega^{(0,3)}](X,Y,Z).
\end{aligned}
\end{align}
where the superscript $(p,q)$ denotes the standard bigrading of forms corresponding to the splitting of the tangent bundle $T\PS=L\oplus \Lt$ given by $K$.
\end{Ex}

\section{Born Geometry}
\label{sec:borngeo}
We have seen that para-Hermitian geometry is a para-complex geometry $(\PS,K)$ with a compatible neutral metric $\eta$, i.e. one of signature $(d,d)$. Now we will take the next step by adding another compatible metric $\HH$ of signature $(2d,0)$ to the picture, giving rise to what has been named Born geometry \cite{Freidel:2014qna}. The additional Riemannian structure $\HH$ defines two more endomorphisms of the tangent bundle which -- along with the para-complex structure $K$ already in place -- form a para-quaternionic structure \cite{Ivanov:2003ze,Freidel:2013zga, Freidel:2015pka}.

\begin{Def}
Let $(\PS,\eta,\omega)$ be a para-Hermitian manifold and let $\HH$ be a Riemannian metric satisfying
\begin{align}
\eta^{-1}\HH=\HH^{-1}\eta,\quad \omega^{-1}\HH=-\HH^{-1}\omega .
\end{align}
Then we call the triple $(\eta,\omega,\HH)$ a Born structure on $\PS$ where $\PS$ is called a Born manifold and $(\PS,\eta,\omega,\HH)$ a Born geometry.
\end{Def}

Note that in this definition we view $(\eta,\omega,\HH)$ as maps $T\PS\to T^*\PS$. In order to understand the nature of a Born structure we review its three fundamental structures. First, as we have seen, it contains an almost {\it para-Hermitian} structure $(\omega, K)$ with compatibility
\be
K^2=\id,\qquad \omega(KX,KY)=-\omega(X,Y). 
\ee
Next, the compatibility between $\eta$ and $\HH$ implies that $J=\eta^{-1}\HH \in \mathrm{End}\, T\PS$ defines what we refer to as a  {\it  chiral} structure\footnote{As we will see, the projectors $P_\pm=\frac12(\id\pm J)$ define a splitting of the tangent space $T\PS = C_+ \oplus C_-$ which mirrors the right/left chiral splitting of string theory. In other words,
given a worldsheet $X:\Sigma \to \PS$ with $\Sigma$ a complex surface, we have that $\pa_z X \in C_+$ and $\pa_{\bar{z}} X \in C_-$. Hence $J$ is called a chiral structure.} $(\eta,J)$ on $\PS$
\begin{align}
J^2=\id,\qquad 
\eta(JX,JY)=\eta(X,Y).
\end{align}
Finally, the compatibility between  $\HH$ and $\omega$ defines an almost {\it Hermitian} structure $(\HH,I)$ on $\PS$
\be
I^2=-\id,\qquad \HH(IX,IY)= \HH(X,Y).
\ee
These three structures, para-Hermitian, chiral and Hermitian satisfy in turn an extra compatibility equation
\be
KJI=\id,
\ee
which follows directly from their definition $K=\omega^{-1}\eta$, $J=\eta^{-1} \HH$, $I=\HH^{-1}\omega$. This means that the triple $(I,J,K)$  form an almost para-quaternionic structure \cite{Ivanov:2003ze,Freidel:2013zga, Freidel:2015pka}
\begin{align}
-I^2=J^2=K^2=\id, \qquad 
\{I,J\}=\{J,K\}=\{K,I\}=0,\qquad 
KJI=\id,
\end{align}
where $\{\ , \ \}$ is the anti-commutator. 

\subsection{The Physics of Born geometry}

Depending on the perspective one wants to emphasize, one can present the Born geometry by selecting two out of the three geometric structures. This gives three different perspectives that allow one to accentuate the relationship of Born geometry with three distinct geometrical structures: Courant geometry, string geometry or the geometry of quantum mechanics.

The emphasis in our paper is on the para-Hermitian triple $(\eta,\omega,K)$ which stresses the relationship of the Born structure with generalized Courant geometry since this structure allows for the construction of a D-bracket $\bracd$. As we have seen, this bracket provides a generalization of the Lie bracket on $M$ where $M$ is viewed as embedded in $\PS$ and $TM=L$ defines a Lagrangian. 

If one is interested in underlining  the relationship of Born geometry with the geometry of string theory, one presents it in terms of its  chiral structure $(\eta, \HH, J)$. The two metrics encode the diffeomorphism and Hamiltonian constraints that arise from the coupling of the string to 2-dimensional gravity. Given the embedding of the closed string $X: S^1\to \PS$ at a given time $\tau$, its dynamics is entirely contained  in the implementation of the constraints \cite{Tseytlin:1990nb,Tseytlin:1990va,Freidel:2013zga,Freidel:2014qna,Freidel:2015pka}
\be\label{St}
\HH(\pa_{\tau} X,\pa_{\tau} X)=0,\qquad \eta(\pa_{\tau} X,\pa_{\tau} X)=0,\qquad \pa_\tau X =J \pa_\sigma X, 
\ee 
where $\sigma$ is the string coordinate. From this perspective the Born structure generalizes the Riemannian geometry $(M,g)$ which is needed to define the dynamics of a relativistic particle $x \in M$. The string equations (\ref{St}) generalize the relativistic constraint $ g(\pa_\tau x,\pa_\tau x)=m^2$ which defines usual relativistic dynamics.


This chiral perspective where Born geometry is defined  as a chiral structure $(\eta,\HH,J)$ augmented by a compatible two-form $\omega$ can be formalized in the following way:
\begin{Lem}
The  Born structure $(\eta,\omega,\HH)$ is equivalent to the chiral triple $(\eta,\HH,J)$ satisfying
\begin{equation}
\begin{aligned}
J&=\eta^{-1}\HH, & J^{2}&=\id,\\
\eta(JX,JY)&=\eta(X,Y), \qquad& \HH(JX,JY)&=\HH(X,Y),
\end{aligned}
\end{equation}
along with a compatible almost symplectic form $\omega$
\begin{align}
\omega^{-1}\eta &= \eta^{-1}\omega,& \omega^{-1}\HH &= -\HH^{-1}\omega.
\end{align}
\end{Lem}
\begin{proof}
We have already shown that starting with a para-Hermitian manifold and a compatible Riemannian metric $\HH$ yields the compatible chiral structure. The converse statement follows from observing that $\omega^{-1}\eta=\eta^{-1}\omega$ implies $K=\eta^{-1}\omega$ is a para-Hermitian structure compatible with $\eta$.
\end{proof}

The chiral structure $J$ plays an important role in subsequent constructions. For this reason, we introduce the chiral subspaces $C_\pm$ which are the $\pm 1$-eigenbundles of $J$ analogous to the eigenbundles $L$ and $\Lt$ of $K$. Together these subspaces span the tangent space of our Born manifold
\begin{equation}
T\PS=C_+\oplus C_-.
\end{equation}
We now also have another set of projection operators onto these subspaces and the corresponding components of vector fields
\begin{equation}\label{eq:chiral-projections}
P_\pm=\frac{1}{2}(\id\pm J) \qandq X_\pm=P_\pm (X).
\end{equation}
The compatibility conditions of $\eta$ and $J$ further imply that the eigenbundles $C_\pm$ are orthogonal with respect to $\eta$
\begin{equation}\label{eq:Cpm-orthogonal-property}
\eta(X_\pm,Y_\mp)=0,
\end{equation}
for any $X_\pm,Y_\pm \in \se(C_\pm)$. The following property that will be used later repeatedly:
\begin{Lem}\label{lem:Kmap-Cpm}
The para-Hermitian structure $K$ on a Born manifold maps the chiral subspaces into each other isomorphically,
\begin{equation}
K:C_\pm\rightarrow C_\mp
\end{equation}
such that $K(X_\pm)=[KX]_\mp$.
\end{Lem}
\begin{proof}
Because $\{J,K\}=0$, $K(X_\pm)=[KX]_\mp$ follows immediately from \eqref{eq:chiral-projections}. Since $K$ is clearly invertible, the map is an isomorphism.
\end{proof}

Finally, it is well-known from works on geometric quantization and the geometry of quantum mechanics that the geometrical structure which underlies  the process of quantization is a Hermitian structure $(\omega,\HH, I)$. In this context, $\PS$ can be understood as a phase space (the phase space of a relativistic particle say), the symplectic structure defines the notion of Hamiltonian flow on the phase space while the complex structure on $\PS$ is naturally associated with coherent states, i.e. states that minimize the uncertainty relation. Conversely as shown in geometric quantization the choice of a K\"ahler structure on $\PS$  uniquely determines an Hilbert space representation in terms of coherent states defined as holomorphic sections over $\PS$.

The key relations between the structures of Born geometry are summarized in Table \ref{tab:BornGeo}.

\begin{table}[h!]
\renewcommand{\arraystretch}{1.7}
\centering
\resizebox{0.95\textwidth}{!}{%
\begin{tabular}{@{}ccc@{}}
\toprule
$I={\HH}^{-1}{\omega}=-\omega^{-1}\HH$ & $J={\eta}^{-1}{\HH}={\HH}^{-1}{\eta}$ & $K={\eta}^{-1}{\omega}={\omega}^{-1}{\eta}$  \vspace{8pt} \\
$-I^2=J^2=K^2=\id$                     & $\{I,J\}=\{J,K\}=\{K,I\}=0$           & $IJK=-\id$     \vspace{8pt} \\
$\HH(IX,IY)=\HH(X,Y)$              & $\eta(IX,IY)= -\eta(X,Y)$         & $\omega(IX,IY)=\omega(X,Y)$             \\
$\HH(JX,JY)=\HH(X,Y)$              & $\eta(JX,JY)= \eta(X,Y)$          & $\omega(JX,JY)=-\omega(X,Y)$            \\
$\HH(KX,KY)=\HH(X,Y)$              & $\eta(KX,KY)= -\eta(X,Y)$         & $\omega(KX,KY)=-\omega(X,Y)$            \\ \bottomrule
\end{tabular}%
}
\caption{Summary of structures in Born geometry.}
\label{tab:BornGeo}
\end{table}

\subsection{Structure Groups}
Each of the three structures $(\eta,\omega,\HH)$ on the manifold $\PS$ corresponds to a structure group on the bundle $T\PS$. The compatibility between the objects then leads to a reduction of the overall structure group of Born geometry to the orthogonal group $O(d)$. The symplectic form $\omega$ comes with structure group $Sp(2d)$ while the metrics $\eta$ and $\HH$ have structure group $O(d,d)$ and $O(2d)$ respectively. 
One of the main reasons behind the definition of Born geometry is the fact that there is a common frame that diagonalizes all  structures simultaneously. Note that in the following we are implicitly using the identification of $\Lt$ with $L^*$ established in Remark \ref{rem:canonicalform}.
\begin{Prop}\label{BG}
The triple $(\eta,\omega,\HH)$ is a Born structure on $\PS$ if and only if there exists a frame $E\in GL(2d)$ such that $(\eta,\omega,\HH)$ can be written as $\eta = E^\T \bar{\eta} E$,  $\omega =E^\T\bar\omega E$ and $\HH=E^\T\bar{\HH} E$ with $(\bar\eta,\bar\omega,\bar\HH)$ given by 
\begin{align}
\bar\eta &= 
\begin{pmatrix}
	0&\id\\
	\id&0
\end{pmatrix} \, , &
\bar\omega &=
\begin{pmatrix}
	0&-\id\\
	\id&0
\end{pmatrix} \, , &
\bar\HH =
\left(
\begin{array}{cc}
	\delta &0\\
	0&\delta^{-1}
\end{array}
\right) \, , \label{m}
\end{align}
where $\delta$ is the constant metric $\delta = \mathrm{diag}( 1,\dots,1)$ on $L$ and its inverse $\delta^{-1}$ is a metric on $\Lt$. Moreover, the frame field $E$ is only determined up to an $O(d)$ transformation $o$: $E \mapsto  \mathrm{diag}(o,o) E$.
\end{Prop}
\begin{proof}
By direct inspection $(\bar\eta,\bar\omega,\bar\HH)$ defines a Born structure and it is easy to check that an $O(d)$ transformation leaves $(\bar\eta,\bar\omega,\bar\HH)$ invariant. Therefore any triple  $(\eta,\omega,\HH) =E^\T(\bar\eta,\bar\omega,\bar\HH)E$ also forms a Born structure.

Conversely, assume $(\eta,\omega,\HH)$ is a Born structure. We will show that there is a frame such that $(\eta,\omega,\HH)$ take the form \eqref{m}. We have noted in Remark \ref{rem:canonicalform} that on the tangent bundle of a para-Hermitian manifold there always exists a frame such that $\eta$ and $\omega$ take the desired form $\bar{\eta}$ and $\bar{\omega}$. This implies that in this frame $K=\text{diag}(\id,-\id)$. Next, we observe that $\HH(KX,KY)=\HH(X,Y)$ implies that $\HH$ is block-diagonal in this frame. Denote the diagonal components by $g(X,Y)=\HH(PX,PY)$ and $\tl{g}(X,Y)=\HH(\tl{P} X,\tl{P} Y)$, i.e. $\HH=\text{diag}(g,\tl{g})$. Writing out the compatibility condition $\HH^{-1}\omega=\omega^{-1}\HH$ explicitly in matrix form yields
\begin{align}
\begin{pmatrix}
0 & -g^{-1} \\
\tl{g}^{-1} & 0
\end{pmatrix}
=
\begin{pmatrix}
0 & -\tl{g} \\
g & 0
\end{pmatrix},
\end{align}
showing that $\tl{g}=g^{-1}$. Now we decompose $g$ into its frame fields, $g=e^\T\delta e$ and use the transformation of the type $E\mapsto \text{diag}(e^{-1},e^\T)E$ (which leaves $\bar{\eta}$ and $\bar{\omega}$ invariant), so that $\HH=\bar{\HH}$. Finally, one can check that the orthogonal change of frame given by $E\mapsto \text{diag}(o,o)E$, where $o^\T\delta o= \delta$, preserves the canonical form of $\bar{\eta},\bar{\omega}$ and $\bar{\HH}$, which completes the proof.
\end{proof}
As an immediate consequence we have the following corollary which gives a clear relationship between Born geometry and para-Hermitian geometry:
\begin{Cor}
Born geometry is equivalent to a choice of metric $g$ on $L$ such that in the splitting $T\PS=L\oplus \Lt$ we have
\begin{align}
\HH=
\begin{pmatrix}
g & 0 \\
0 & g^{-1}
\end{pmatrix}.
\end{align}
\end{Cor}
This corollary expresses mathematically one of the key ideas behind the Born geometry interpretation of string geometry. It essentially stipulates that one can interpret the Born geometry as a para-Hermitian structure together with a ``Riemannian'' metric on its Lagrangian. In this interpretation, the $\HH$ metric is always diagonal and the rest of the geometry that includes fluxes is encoded into the choice of the bracket.

A way to understand the result of Proposition \ref{BG} is to look at the pairwise intersection of the structure groups $O(2d), O(d,d)$ and $Sp(2d)$. On a para-Hermitian manifold with compatible $\eta$ and $\omega$ the structure group is 
\begin{equation}
  Sp(2d) \cap O(d,d) = GL(d)
\end{equation}
which corresponds to the real structure $K=\eta^{-1}\omega$ providing the splitting into $d$-dimensional subspaces. One can think of this intersection as the structure group of general relativity, i.e. the change of frames. The embedding  $k: GL(d)\to Sp(2d) \cap O(d,d)$ is explicitly given by 
\be
 k(e) = \left(
\begin{array}{cc}
	e&0\\
	0&(e^\T)^{-1}
\end{array}
\right).
\ee
The symplectic form $\omega$ together with the compatible metric $\HH$ defines a complex structure $I=\HH^{-1}\omega $ and we have
\begin{equation}
  Sp(2d) \cap O(2d) = U(d)
\end{equation}
which can be seen as the structure group of quantum mechanics where the complex structure $I$ facilitates the Born reciprocity $x\rightarrow p, p\rightarrow -x$. 
The explicit embbeding $\imath: U(d) \to Sp(2d) \cap O(2d)$ is given by 
\be
\imath(A+iB) =   \left(
\begin{array}{cc}
	A&B\delta^{-1}\\
	-\delta B& \delta A\delta^{-1}
\end{array}
\right).
\ee
Note that orthogonality is defined with respect to flat metric $\delta$ in \eqref{m}: $A^\T \delta A = \delta$, whereas the unitarity condition is
\be
(A-iB)^\T \delta (A+i B)= \delta,
\ee
i.e. $U=A+iB$ preserve the sesquilinear form $\delta$. 
Further note that one also has $Sp(2d) \cap GL(d,\mathbb{C}) = U(d)$ and $O(2d) \cap GL(d,\mathbb{C}) = U(d)$ which is the ``2-out-of-3'' property of the unitary group, i.e. any two of these three groups intersect to give the unitary group and imply the presence of the third group. This corresponds to the compatibility of the triple $(\omega,\HH,I)$.

If one only considers the compatible two metrics $\eta$ and $\HH$, one obtains the string symmetry group
\begin{equation}
  O(d,d) \cap O(2d) = O(d)\times O(d),
\end{equation}
with explicit embedding $\jmath: O(d)\times O(d) \to  O(d,d) \cap O(2d)$ given by 
\be
\jmath(\OO_+,\OO_-)=  \frac12 \left(
\begin{array}{cc}
	(\OO_++\OO_-)&(\OO_+-\OO_-)\\
	(\OO_+-\OO_-)& (\OO_++\OO_-)
\end{array}
\right),
\ee
where $\OO_\pm\in O(d)$, i.e. $\OO_\pm^\T\delta \OO_\pm=\delta$. The symmetries $\OO_\pm$ correspond to individual Lorentz groups associated to the right and left moving sectors of the closed string which are respectively in the image of $P_\pm=\frac12(\id \pm J)\in C_\pm$.

Finally we can look at the intersection of all three groups giving the structure group of Born geometry
\begin{equation}
  Sp(2d) \cap O(d,d) \cap O(2d) = O(d) .
\label{eq:structuregroup}
\end{equation}
This is of course the structure group of Riemannian geometry. 
This shows that the structures of general relativity, quantum mechanics and string theory all appear in Born geometry with the Lorentz group as their common subgroup.

\section{The Born Connection}
We are now ready to formulate and prove our main theorem on the existence and uniqueness of a connection for Born geometry. We start with the definition of compatibility.
\begin{Def}
Let $(\PS,\eta,\omega,\Hh)$ be a Born geometry. A connection $\n$ that preserves the three defining structures of the Born geometry,
\begin{equation}
\n\eta = \n\Hh = \n\omega = 0,
\end{equation}
will be called {compatible with the Born geometry} $(\PS,\eta,\omega,\Hh)$.
\end{Def}

\subsection{Main Theorem}
Now we can concisely state our main  theorem:
\begin{Thm}\label{th:BornConnection}
Let $(\PS,\eta,\omega,\Hh)$ be a Born geometry with $K=\eta^{-1}\omega$ the corresponding almost para-Hermitian structure. Then there exists a unique connection called the Born connection denoted by $\nB$ which
\begin{itemize}
  \item is compatible with the Born geometry $(\PS,\eta,\omega,\Hh)$,
  \item has a vanishing generalized torsion $\Th=0$.
\end{itemize}
This connection can be explicitly expressed in terms of the three defining structures $(\eta,\omega,\Hh)$ of the Born geometry and the canonical D-bracket. It can be concisely written in terms of the para-Hermitian structure $K$ and the chiral projections $X_\pm=P_\pm(X)$ as
\begin{equation}
\nB_XY = \bl X_-,Y_+ \br_+ + \bl X_+,Y_-\br_- + (K\bl X_+,KY_+ \br)_+ + (K\bl X_-,KY_-\br)_- .
\label{eq:BornBracket}
\end{equation}
\end{Thm}

The Born connection can be expressed in various different forms. Two useful ones are derived in the appendix and given here. It is important to note that the notation we use means $KX_\pm:=K(X_\pm)$  which is an element of $C_\mp$, cf. Lemma \ref{lem:Kmap-Cpm}.
\begin{Cor}
By contracting with a third vector $Z$ and using \eqref{eq:etaKK}, the concise expression \eqref{eq:BornBracket} for the Born connection in terms of projected components  and the bracket can be written as 
\begin{equation}
\begin{aligned}
\eta(\nB_XY,Z) &= \eta(\bl X_-,Y_+ \br_,Z_+)  + \eta( \bl X_+,Y_-\br,Z_-) \\
	&\qquad  - \eta( \bl X_+,KY_+\br, KZ_+) - \eta(\bl X_-,KY_-\br) ,KZ_-) 
\end{aligned}
\label{eq:BornExpanded1}
\end{equation}
and subsequently expanded explicitly in terms of $(I,J,K)$ as
\begin{equation}
\begin{aligned}
\eta(\nB_XY,Z)&= \frac14\Big\{ 
	\eta\big( \bl X,Y \br - \bl JX,JY \br, Z\big) + \eta\big( \bl X,JY \br - \bl JX,Y \br, JZ\big) \\
	&\qquad - \eta\big( \bl X,KY \br - \bl JX,IY \br, KZ\big) - \eta\big( \bl X,IY \br - \bl JX,KY \br, IZ\big) \Big\}  .
\end{aligned}
\label{eq:BornExpanded2}
\end{equation}
\end{Cor}

\begin{Cor}
\label{cor:BornContorsion}
The Born connection can be expressed as the canonical connection $\nc$ plus a contorsion term $\OB$ 
\begin{equation}
\eta(\nB_XY,Z) = \eta(\nc_XY,Z) + \OB(X,Y,Z)
\label{eq:BornContorsion}
\end{equation}
which can be expanded in terms of derivatives of $\HH$ and $(I,J,K)$ as
\bea
\OB(X,Y,JZ) &:=& \frac{1}{2}\nc_X \HH(Y,Z) \\
&+&  \frac{1}{4}\Big[\nc_Y\Hh(Z,X)+\nc_{KY}\Hh(KZ,X) 
+\nc_{JY}\Hh(JZ,X)+\nc_{IY}\Hh(IZ,X)\Big] \cr
&-&\frac{1}{4}\Big[\nc_{Z}\Hh(Y,X) + \nc_{KZ}\Hh(KY,X)+ \nc_{JZ}\Hh(JY,X)  + \nc_{IZ}\Hh(IY,X) 
 \Big] .\nonumber
\label{OB}
\eea

\end{Cor}

\subsection{Proof}
\label{sec:proof}
We will now prove the theorem and illuminate some of its properties along the way. We will also derive an expression for its generalized torsion (which subsequently vanishes). We begin by proving the existence of a connection which has the properties of the Born connection, i.e. it is compatible with the three structures of Born geometry and has vanishing generalized torsion. This is followed by a proof of uniqueness and we have a closer look at the torsion.

\subsubsection{Proof of Existence}
In this section we show that the expression \eqref{eq:BornBracket} indeed defines a connection and has the listed properties, i.e. it is compatible with all the structures of Born geometry and its generalized torsion vanishes. This expression is very similar to the construction of the Bismut connection given by Gualtieri in \cite{Gualtieri:2007bq}.
Theorem \ref{th:Born-LCconnection} will explain a posteriori why.

\paragraph*{Tensoriality and Leibniz property}
First, we start by showing that \eqref{eq:BornBracket} defines a connection on $T\PS$. Using the properties of the D-bracket,
\begin{equation}
\begin{aligned}
\nB_{fX}Y &= f\nB_XY -Y_+[f](X_-)_+-Y_-[f](X_+)_-\\
&\qquad -K(Y_+)[f](KX_+)_+- K(Y_-)[f](KX_-)_- \\
&\qquad  +\eta(X_-,Y_+)(\mathcal{D}f)_++\eta(X_+,Y_-)(\mathcal{D}f)_- \\
&\qquad + \eta(X_+,KY_+)(K\mathcal{D}f)_++\eta(X_-,KY_-)(K\mathcal{D}f)_-\\
&=f\nB_XY,
\end{aligned}
\end{equation}
where we made use of Lemma \ref{lem:Kmap-Cpm} along with the fact that $C_\pm$ are orthogonal with respect to $\eta$ \eqref{eq:Cpm-orthogonal-property} and complementary to each other. Similarly, we have
\begin{equation}
\begin{aligned}
\nB_{X}fY &= f\nB_XY +X_-[f](Y_+)_++X_+[f](Y_-)_- \\
&\qquad +X_+[f](K^2Y_+)_++X_-[f](K^2Y_-)_-\\
&=f\nB_XY+X[f]Y,
\end{aligned}
\end{equation}
which shows that $\nB$ indeed is a connection. Next, we 
will show the compatibility with the various Born structures.
\paragraph*{Compatibility with Born geometry} We first show the compatibility with $\eta$. It follows directly from the  compatibility \eqref{eq:DbracketComp} condition for the D-bracket
\be
X[\eta(Y,Z)] = \eta(\bl X,Y\br,Z)+\eta(Y,\bl X,Z\br).
\ee
Using this and the expansion \eqref{eq:BornExpanded1} one obtains that 
\begin{equation}
\begin{aligned}
\eta(\nB_XY,Z)+\eta(Y,\nB_XZ) &= X_-[\eta(Y_+,Z_+)]+X_+[\eta(Y_-,Z_-)]  \\
	&\qquad -X_+[\eta(KY_+,KZ_+)]-X_-[\eta(KY_-,KZ_-)] \\
	&= X[\eta(Y,Z)],
\end{aligned}
\end{equation}
The last equality follows from  $\eta(KX,KY)=-\eta(X,Y)$ and $\eta(X,Y)=\eta(X_+,Y_+) + \eta(X_-,Y_-)$. Next, we show the compatibility with the chiral structure $J$. This is equivalent to showing that $\nB$ preserves its eigenbundles, i.e. $\nB_XY_\pm\in \se(C_\pm)$. This property is obvious from the expression \eqref{eq:BornBracket}, since
\begin{align}
\nB_XY_\pm=\bl X_\mp,Y_\pm\br_\pm+(K\bl X_\pm,KY_\pm\br)_\pm \in \se(C_\pm).
\end{align}
Lastly, we show the compatibility with $K$. For this we again use Lemma \ref{lem:Kmap-Cpm} to write
\begin{align}
\nB_X(KY)=\bl X_-,KY_- \br_+ + \bl X_+,KY_+\br_- + (K\bl X_+,Y_-\br)_+ + (K\bl X_-,Y_+\br)_-,
\end{align}
and
\begin{align}
K(\nB_XY) = (K\bl X_-,Y_+ \br)_- + (K\bl X_+,Y_-\br)_+ + \bl X_+,KY_+\br_- + \bl X_-,KY_-\br_+,
\end{align}
so that $(\nB_XK)Y=\nB_X(KY)-K(\nB_XY)=0$. Now, since $\omega=\eta K$ and $\HH=\eta J$, the compatibility with the remaining structures of Born geometry follows.

\paragraph*{Vanishing generalized torsion}
Here we compute the generalized torsion of the Born connection and show that it indeed vanishes. We begin with the following lemma which is proven in the appendix
\begin{Lem}\label{lem:LemmaForTorsion}
The D-bracket of chiral vectors $X_\pm=P_\pm(X)$ satisfies
\begin{align}
\eta(\bl X_\mp,Y_\pm\br,Z_\pm) &= - \eta(\bl Y_\pm,X_\mp\br,Z_\pm) \\
\eta(\bl X_\pm,Y_\pm\br,Z_\mp) &= \eta(\bl Z_\mp,X_\pm\br,Y_\pm) .
\end{align}
\end{Lem}

Recalling the definition of generalized torsion \eqref{eq:gen-torsion} as the difference between the brackets $\bracd^\n$ and $\bracd$, the generalized torsion for the Born connection is given by
\begin{equation}
\Th(X,Y,Z) = \eta(\nB_XY-\nB_YX,Z)+\eta(\nB_ZX,Y)-\eta(\bl X,Y\br,Z) .
\label{eq:torsioncalc1}
\end{equation}
We use the expanded form  of the Born connection  \eqref{eq:BornExpanded1} to rewrite the second and third term in the generalized torsion. Making use of  Lemmas \ref{lem:Kmap-Cpm} and \ref{lem:LemmaForTorsion} we have
\begin{align}
-\eta(\nB_YX,Z) &= -\eta(\bl Y_-,X_+ \br_,Z_+)  - \eta( \bl Y_+,X_-\br,Z_-) \notag\\
	&\qquad  + \eta( \bl Y_+,KX_+\br, KZ_+) + \eta(\bl Y_-,KX_-\br) ,KZ_-) \notag\\
	&= \eta(\bl X_+,Y_- \br_,Z_+)  + \eta( \bl X_-,Y_+\br,Z_-) \notag\\
	&\qquad  - \eta( \bl KX_+,Y_+\br, KZ_+) - \eta(\bl KX_-,Y_-\br) ,KZ_-) \label{eq:torsioncalc3}\\
\eta(\nB_ZX,Y) &= \eta(\bl Z_-,X_+ \br_,Y_+)  + \eta( \bl Z_+,X_-\br,Y_-) \notag\\
	&\qquad  - \eta( \bl Z_+,KX_+\br, KY_+) - \eta(\bl Z_-,KX_-\br) ,KY_-) \notag\\
&= \eta(\bl X_+,Y_+ \br_,Z_-)  + \eta( \bl X_-,Y_-\br,Z_+) \notag\\
	&\qquad  - \eta( \bl KX_+,KY_+\br, Z_+) - \eta(\bl KX_-,KY_-\br)  Z_-,). \label{eq:torsioncalc4}
\end{align}
We now add all the terms in \eqref{eq:torsioncalc1} and after some rearranging find
\begin{equation}
\begin{aligned}
\Th(X,Y,Z)=-\sum_\pm\Big(&\eta(\bl X_\pm,Y_\pm \br,Z_\pm)+\eta(\bl X_\pm,KY_\pm\br ,KZ_\pm)\\
+&\eta(\bl KX_\pm,Y_\pm\br,KZ_\pm)+\eta(\bl KX_\pm,KY_\pm\br ,Z_\pm)\Big)
\end{aligned}
\end{equation}
which in turn can be written in terms of the  generalized Nijenhuis tensor \eqref{eq:gen-nij-def} of $K$
\begin{align}\label{eq:gen-torsion-nijenhuis}
\Th(X,Y,Z)=
-\sum_\pm&\eta(\bl X_\pm,Y_\pm \br+\bl KX_\pm,KY_\pm\br  
	- K\bl KX_\pm,Y_\pm\br - K\bl X_\pm,KY_\pm\br ,Z_\pm) \notag\\
	=-\sum_\pm& \eta({\cal N}_K(X_\pm,Y_\pm),Z_\pm).
\end{align}
Since $\bracd$ is a D-bracket (Definition \ref{def:Dbracket}) we have $\NN_K=0$ and this vanishes. Therefore, we find that the generalized torsion of the Born connection is indeed zero
\begin{equation}
\Th(X,Y,Z) = 0 \, .
\end{equation}

\subsubsection{Proof of Uniqueness}
In this section we use the properties of the Born connection to show that it has to take the form \eqref{eq:BornBracket}. 
We start by assuming that $\tilde{\n}$ is any connection that satisfies the condition of our theorem: it preserves the Born structure, in particular it satisfies $\tilde{\n} J= \tilde{\n}K=0$, and has vanishing torsion.
It is convenient to evaluate independently  each of the four individual chiral components
\begin{equation}
\eta(\tilde{\n}_XY,Z) = \sum_\pm\ \big[ \eta(\tilde{\n}_{X_\pm} Y_\mp,Z) + \eta(\tilde{\n}_{X_\pm} Y_\pm,Z) \big] .
\end{equation}
We start with the evaluation of $\tilde{\n}_{X_+}Y_-$
\begin{align}
\eta(\tilde{\n}_{X_+}Y_-,Z)&\overset{\tilde{\n} J=0}{=}\eta(\nB_{X_+}Y_-,Z_-)\overset{\tilde\Th=0}{=}\eta(\bl X_+,Y_-\br,Z_-)+\eta(\nB_{Y_-}X_+,Z_-)-\eta(\nB_{Z_-}X_+,Y_-) \notag\\
&\overset{\tilde{\n} J=0}{=}\eta(\bl X_+,Y_-\br,Z_-)=
\eta(\bl X_+,Y_-\br_-,Z),
\end{align}
where we used $\eta(X_\pm,Y_\mp)=0$. Therefore, we conclude that the mixed components take the form $\tilde{\n}_{X_\pm}Y_\mp=\bl X_\pm,Y_\mp\br_\mp$. For the remaining components we calculate
\begin{equation}
\begin{aligned}
\eta(\tilde{\n}_{X_+}Y_+,Z)&\overset{\tilde{\n} J=0}{=}\eta(\tilde{\n}_{X_+}Y_+,Z_+)\overset{\tilde{\n} K=0}{=}-\eta(\tilde\n_{X_+}KY_+,KZ_+)\\
&\overset{\tilde\Th=0}{=}-\eta(\bl X_+,KY_+\br,KZ_+)-\eta(\tilde\n_{KY_+}X_+,KZ_+)+\eta(\tilde\n_{KZ_+}X_+,KY_+)\\
&\overset{\tilde{\n} J=0}{=}-\eta(\bl X_+,KY_+\br,KZ_+)=\eta(K(\bl X_+,KY_+\br)_+,Z),
\end{aligned}
\end{equation}
where we used Lemma \ref{lem:Kmap-Cpm} and that $\eta(X_\pm,KY_\pm)=0$. Therefore we have for the other components $\nB_{X_\pm}Y_\pm = (K\bl X_\pm,KY_\pm\br)_\pm$. Putting everything together, this shows that $\tilde{\nabla}$ is uniquely determined by the listed conditions and therefore equal to  $\nB$.

\section{Relationship to Generalized Geometry and DFT}

In this final section we will look at the Born connection in the context of Generalized Geometry and Double Field Theory (DFT). We now have the tools to give precise relations between the various objects in each framework. This will be spelled out in more detail in the future work \cite{Freidel:new}.

\subsection{From D-bracket to Courant algebroids }
Para-Hermitian geometry is closely related to generalized geometry \cite{Hitchin:2004ut,Gualtieri:2003dx}, i.e. the geometry of the generalized tangent bundle $\mathbb{T}M\coloneqq TM\oplus T^*M$. This has been observed in \cite{Vaisman:2012px,Vaisman:2012ke} and explored in more detail in \cite{Freidel:2017yuv,Svoboda:2018rci}. We will briefly recall this relationship here and relate the notions of a Courant bracket and generalized torsion defined on $\Tt M$ to our geometrical set-up of a tangent bundle $T\PS$ equipped with the D-bracket.

Generalized geometry is based on the concept of a transitive Courant algebroid \cite{Xu} which is encoded into the quadruple $(E, \la\,,\,\ra_E, [\,,\,]_E,\pi)$ where $E\to M$ is a vector bundle over $M$ equipped with an invertible fiber metric $\la\,,\,\ra_E$ and a surjective anchor $\pi: E\to TM$. The basic axioms are that  $(E, \la\,,\,\ra_E, [\,,\,]_E,\pi)$ is a metric algebroid which includes metric compatibility, the Leibniz property and normalisation
\bea
\pi(u)\la v, w\ra_E &=& \la [u,v]_E, w\ra_E+\la v,[u,w]_E\ra_E, \label{eq:DbracketCompalg} \\
{[}v, f u]_E &=& f[v,  u]_E+\pi(v)[f]u, \\ 
 \la {[}u,v]_E, v\ra_E&=&  \la u,{[}v,v]_E \ra_E\label{eq:DbracketNormalg}
\eea
where $u,v,w\in\Gamma(E)$ (cf. Definition \ref{def:MetricAlgebroid} for a metric algebroid on $T\PS$ with identity anchor). We can now define the algebroid differential $\cD_\pi$ by $\la \mathcal{D}_\pi f, V\ra_E= \pi(V)[f]$.
In order to obtain a Courant algebroid one demands in addition that the metric compatible bracket satisfies the Jacobi identity
\begin{equation}
[ u,[ v,w]_E]_E = [[ u,v]_E,w]_E + [v,[ u,w]_E]_E,
\end{equation}
which means that each fiber of $E$ forms a Leibniz algebra as introduced by Loday \cite{Loday,Lodayreview}. The Leibniz identity implies $\pi ([ u,v]_E) = [\pi(u),\pi(v)]$, stating that $\pi$ is a morphism intertwining the Lie bracket with the metric algebroid bracket. An important subclass of transitive Courant algebroids are the ones which are exact\footnote{It is often assumed in addition that the exact sequence $0\to \mathrm{Ker}\,\pi\to E \to TM\to0$ is split. We would rather call these split algebroids rather than exact.}, i.e. the ones for which  $\mathrm{Ker}\,\pi =(\mathrm{Ker}\,\pi)^\perp$. 

The canonical  example of such a structure on any manifold $M$ is the \textit{standard Courant algebroid} $(\Tt M,\la\ ,\ \ra,\brac,\pi)$. Here $\Tt M=TM\oplus T^*M$ is the doubled bundle, its anchoring map is given by the canonical projection $\pi:\Tt M\rightarrow TM$, the Dorfman bracket $\brac$ reads 
\be\label{Standard}
[x+\ap,y+\bt]=[x,y]+L_x\bt-\imath_y \rd\ap,
\ee
where $[x,y]$ denotes the Lie bracket and $L_x$ the Lie derivative, and the canonical pairing is
\be\label{eq:pairingTT}
\la x+\ap,y+\bt\ra=\imath_x\bt+\imath_y\ap,
\ee
for $x, y \in TM $ and $\ap,\bt \in T^* M$. The standard bracket can be twisted by a closed three-form $H$ called the H-flux to define a non-standard bracket (or twisted bracket) $\brac_H$ given by 
\be
[ x+\ap,y+\bt]_H = [ x+\ap,y+\bt] + \imath_x\imath_y H.
\ee 
The Jacobi identity for the non-standard bracket is equivalent to the condition $\rd H=0$ \cite{Severa:2001qm}. 

In the case where the para-Hermitian structure is $L$-integrable there is a direct correspondence between the para-Hermitian geometry with D-structures and generalized geometry (or exact Courant algebroids):
\begin{Thm}\label{th:Dstructure-Calgebroid}
Any  D-structure which is based on an $L$-integrable para-Hermitian manifold gives rise to an exact Courant algebroid whose anchor map is the projection $P$. Moreover, the Courant algebroid obtained from the canonical D-structure is canonically isomorphic to the standard Courant algebroid.   
\end{Thm}
A version of this central result was first established in \cite{Freidel:2017yuv} and the relationship between $L$-integrable structures and Courant alegbroid  was further studied in \cite{Svoboda:2018rci}. We will devote a subsequent paper to a deeper study of this correspondence using the framework of D-structures briefly introduced in Section \ref{sec:Dstructure}. This result is key to the relationship between para-Hermitian geometry and generalized geometry which has been first studied by Vaisman in \cite{Vaisman:2012px,Vaisman:2012ke}.

We now describe this relationship in the restricted context of the canonical D-bracket. One first assumes that  $L$ is integrable. In this case $\PS$ can be viewed as a foliated manifold $\ell$ where $\ell$ denotes the union of all leafs -- called the integral foliations of $L$ -- and as such is a half-dimensional manifold (compared to $\PS$). By definition, $L$ is then the tangent bundle of this manifold, $L=T\ell$. Then $T \PS$ can then be viewed as a bundle over $\ell$ with anchoring map simply given by the para-Hermitian projection $P: T\PS \to T\ell$. 

To go further, one first notices that the canonical D-bracket can be expressed as the sum 
\be\label{Lsplit}
\bl X, Y\br= \bl X, Y\br_L +\bl X,Y\br_{\Lt}
\ee
where the $L$-bracket $\bl X, Y\br_L$ contains only derivatives along $L$ and is given by
\begin{equation}
\eta(\bl X,Y\br_L,Z)=\eta(\nc_{PX}Y-\nc_{PY}X,Z)+\eta(\nc_{PZ}X,Y).
\end{equation}
The $L$-bracket is characterised by three defining properties
\be\label{defining}
 \bl \Pt X,\Pt Y \br_L =0,\qquad P(\bl P X ,\Pt Y \br_L)=0, \qquad
 \bl X ,f Y \br_L = f  \bl X ,Y \br_L +  P X[f] Y .
\ee
The first one states that the $L$-bracket vanishes on $\Lt$, while the second means that the action of $L$ on $\Lt$ stays in $\Lt$, the third one implies that $P$ is the anchor for the $L$-bracket. To make the notation less cluttered, we introduce
\begin{align}
x:= PX \in L , \quad \xt:=\Pt X\in \tilde{L}.
\end{align}
The metric $\eta$ can be viewed as a map $\eta: T\PS\to T^*\PS$, allowing us to identify ${\tilde{L}}$ with $L^*= T^*\ell$ and to define the following isomorphism 
\begin{equation}\label{iso}
\begin{aligned}
\rho:    T\PS&\rightarrow \Tt \ell \\
X= x+\xt  &\mapsto x+\alpha \quad\mathrm{with}\quad \alpha =\eta (\xt).
\end{aligned}
\end{equation}
Given this isomorphism, we can define the notion of a Lie derivative along $L$ acting on $\Lt$ by
\begin{align}
\Lie_x \yt \coloneqq \rho^{-1}L_x\rho(\yt),
\end{align}
where $L_x$ is the usual Lie derivative of a one-form in $L^*$ along $L$. Equivalently, this can be written as
\be
\eta (\mathcal{L}_x \yt, z) = x [\eta(\yt, z)] - \eta(\yt, [x,z]).
\ee

\begin{Lem}\label{Lbrac}
The $L$-bracket is explicitly given by 
\be\label{bracketexpl}
\bl x+\tilde{x} ,y+\tilde{y} \br_L =
[x,y] + \mathcal{L}_x \yt - \mathcal{L}_y \xt + \Db \eta(x,\yt) - N(x,y),
\ee where $[\,,\,]$ is the Lie bracket on $T \PS$, $\Db$ is the projected differential, $\Db f \in \Lt$ is given by $\eta(Z, \Db f) = z[f]$ and $N=\Pt N_K$ is the projected Nijenhuis tensor. 
\end{Lem}
The proof of this important lemma is mechanical, the main aspects are  already spelled out in \cite{Freidel:2017yuv}, it simply uses the definition of the Lie derivative and the fact that the torsion of the canonical connection is polarized and given by the projected Nijenhuis tensor (see Lemma \ref{polarised}). It is  important to note that this expression is valid even if $L$ is non-integrable. The bracket $[x,y]$ used in the expression is the Lie Bracket on $T \PS$ and the combination  $[x,y]- N(x,y)$ is the projection of $[x,y]$ along $L$. 

When $L$ is integrable, $N$ vanishes, $[x,y]$ is in $L$, and we can prove the first claim of the theorem: the $L$-bracket satisfies the Leibniz algebroid axioms and the projection map $P:X\to x $ of $T \PS$ onto $L$ satisfies the anchoring property.
It is also interesting to note that the set $\mathbb{Z}$ of elements  of the form $\Db f$ form a ideal for the $L$-bracket
\be\label{ideal}
\bl \Db f , Y \br_L =0,\qquad 
\bl X , \Db f  \br_L = \Db x[f].
\ee
We see more precisely that $\mathbb{Z}$ is in the left center of the bracket, $\bl \mathbb{Z},T\PS \br_L=0$  while it is a right ideal $\bl T\PS, \mathbb{Z} \br_L\subset \mathbb{Z}$.

We can finally show that it satisfies the Jacobi identity. The Jacobiator which measures the failure to satisfy Jacobi is given by $J(X,Y,Z):= \bl \bl X,Y\br_L, Z\br_L-\bl X, \bl Y, Z\br_L\br_L + \bl Y, \bl X, Z\br_L\br_L$ and we can look at its different components. The detailed proof is in the appendix.
%
 This proves the first claim of the theorem: the $L$-bracket satisfy the Leibniz algebroid axioms. 
 
The second part of the theorem claims that the Courant algebroid associated with the $L$-bracket is isomorphic to the standard Courant algebroid \cite{Freidel:2017yuv,Svoboda:2018rci}. This now follows easily from Lemma \ref{Lbrac} and the isomorphism \eqref{iso} which identifies ${\tilde{L}}$ with $L^*= T^*\ell$. This isomorphism defines a Courant algebroid isomorphism between $(T\PS,\eta,\bracd_L,P)$ and the {standard Courant algebroid} $(\Tt \ell,\la\ ,\ \ra,\brac_\ell,\pi)$. It maps $P$ onto $\pi$ by construction  and it intertwines the brackets 
\begin{equation}\label{eq:D-brac-dorfman}
\rho \bl X,Y\br_L =  [\rho X,\rho Y]_\ell. 
\end{equation}
 
Let us finally note  that the definition of the  D-bracket  does not require the underlying para-Hermitian structure $(K,\eta)$ to be integrable and the integrability is only needed to make contact with generalized geometry. If the the para-Hermitian structure $(\PS,\eta,K)$ is not integrable, we expect that we can still make contact with generalized geometry by allowing non-trivial fluxes \cite{Freidel:new}.

\subsection{From Courant algebroids to generalized connections}
One of the main interests in generalized geometry from the mathematical side stems from the fact that we can define connections, torsion, Riemannian-like  and Ricci-like notions of curvature and ultimately a generalization of the Einstein-Hilbert action. 

From the physics side the interest is mainly that it is then possible to rewrite the one-loop beta function equation of a wealth of string theories simply as the vanishing of the generalized Ricci tensor \cite{StreetsJ}. In other words, for many string models the Zamolodchikov action governing the string renormalization is a generalized form of the Einstein equations.

We now briefly review these notions before exploring the connection they entertain with the para-Hermitian geometry. We refer the reader to  \cite{Jurco:2016emw} for an excellent review of these geometrical aspects.
Let's start by describing the generalized geometry associated with the Courant algebroid given by $(E, \la\,,\,\ra_E, [\,,\,]_E,\pi)$. A Courant algebroid connection $D_X$ associated with an anchor $\pi$ is a derivation of $\Gamma(E)$ that satisfies the Leibnitz rule
\be
D_{fX} Y = f D_XY,\qquad D_X fY = f D_X Y + \pi(X)[f] Y,
\label{eq:gen-con-prop}
\ee
and preserves the pairing 
\be
\pi(X)\la X,Y\ra_E = \la D_XY,Z\ra_E +\la Y,D_X,Z\ra_E.
\ee
Given such a connection we can define its generalized torsion $\Th_D(X,Y)$ with coefficients $ \Th_D(X,Y,Z) = \la \Th_D(X,Y), Z\ra_E$ to be given by \cite{Gualtieri:2007bq,Coimbra:2011nw}
\be
\Th_D(X,Y,Z) := \la D_XY -D_YX, Z\ra_E + \la D_ZX,Y\ra_E - \la[X,Y]_E, Z\ra_E.
\ee
It can be checked that this is tensorial.
We can also define the Riemann curvature operator, first  introduced in \cite{Hohm:2012mf}. The formula is more involved
\bea
2 {\cal R}(X,Y,Z,W) &:=&
\la([D_X,D_Y]Z -D_{[X,Y]_E}Z),W\ra_E \cr
&+& \la([D_Z,D_W]X -D_{[Z,W]_E}X),Y\ra_E \cr
&+ &  \la D_a X, Y\ra_E \la D^a Z, W \ra_E
\eea 
where $D_a$ and $D^a$ denote the covariant derivative with respect to a basis and its dual respectively. It can be checked that this is a tensor.

A Courant algebroid can be equipped with a generalized metric $\HH$ which is usually chosen to be Euclidean. By definition, a generalized metric is characterised by the choice of a chiral structure $J$: $\HH(X,Y)=\la J(X),Y\ra_E$ where $J$ is a bundle isomorphism which is compatible with the pairing, $\la JX,JY\ra_E=\la X,Y\ra_E$ and squares to $\id$. 

Recall that the dual of the anchor map $\pi$ is an embedding map $ \pi^*: T^* M  \hookrightarrow \Tt M$. It turns out that the presence of a generalized metric $\HH$ on an exact Courant algebroid gives a canonical splitting $e: TM \hookrightarrow \Tt M$ of the exact sequence
\begin{align}
0\longrightarrow T^*M \overset{\pi^*}{\longrightarrow} \Tt M \overset{\pi}{\longrightarrow} TM \longrightarrow 0,
\end{align}
satisfying $\pi \circ e = \id$, $ \la e(x),e(y)\ra_E =0$, for $x,y \in TM$. In such splitting the metric $\HH$ is block diagonal, $\HH={\rm{diag}}(g,g^{-1})$, $g$ being a usual metric on $TM$. It is easy to check that this map exists and is unique.\footnote{It can be explicitly constructed using the dual of the anchor map $\pi^*:T^*M\to \Tt M$. First one uses that $g^{-1}= \pi J\pi^*: T^* M\to TM$ defines a metric and then that $e=  J \pi^* g$ has all the required properties \cite{StreetsJ}.}

Now that we have the generalized metric $\HH$ on $E$, we can construct the generalized ``Levi-Civita'' connection which preserves $\HH$ and is torsionless in the generalized sense. Unlike the case of usual geometry where there is a unique Levi-Civita connection, these conditions {\it do not} determine the generalized connection uniquely. The difference $\chi =\nabla -\nabla'$ between any two Levi-Civita connections is given by a tensor which is totally skew and chiral
\be
\chi(X,Y,Z) = \chi(X_+,Y_+,Z_+) +  \chi(X_-,Y_-,Z_-),
\ee
where $X_\pm =\tfrac12 (\id\pm J) X$. Although there is no unique Levi-Civita connection, one can resort to a preferred Levi-Civita connection instead. In the context of the exact Courant algebroid described above in (\ref{Standard}), equipped with the diagonal generalized metric $\HH= {\mathrm{diag}}(g,g^{-1})$, the  preferred Levi-Civita connection $D_x$ on $\Tt M$ is simply given by the Levi-Civita connection $\nabla^g$ of $g$ on $TM$ and $T^*M$
\be\label{connection}
D_x (y+\alpha) = \nabla^g_x y +  \nabla^g_x \alpha.
\ee 
This structure and its extension to the case of a non-standard Courant  algebroids was first revealed by Ellwood \cite{Ellwood:2006ya}.

\subsection{From generalized connections to the Born connection}

Given Theorem \ref{th:Dstructure-Calgebroid}, it should now be clear to the reader that there is a natural correspondence between para-Hermitian geometry and the Born connection on one side and the generalized Levi-Civita connection on the other. The goal of this section is to spell out this correspondence in detail and show how the background para-Hermitian structure on $\PS$ allows us to construct a generalized connection from a regular connection.

We assume here that $L$ is integrable. The Born connection has been explicitly defined in (\ref{eq:BornBracket}) in terms of the D-bracket and since the D-bracket naturally splits into the sum of an  $L$-bracket plus an $\Lt$-bracket we can therefore also split $ \nB_X = \DB_X +\DBt_X $ where $\DB_X$ involves only the $L$-bracket
\be
\DB_XY = (\bl X_-,Y_+ \br_{L})_+ + (\bl X_+,Y_-\br_{L})_- + (K\bl X_+,K(Y_+) \br_L )_{+} + (K\bl X_-,K(Y_-)\br_L)_{-} .
\label{eq:BornBracketL}
\ee
It can be checked that  $\DB$ is a generalized connection associated with the anchor $P$. It is obtained from the full Born connection by projecting the bracket and we will call it the projected Born connection. The proof that it satisfies the axioms (\ref{eq:gen-con-prop}) is similar to the proof  for the full Born connection. In addition one has to use that the $L$-bracket is constructed in terms of the partial connection $\nc_{PX}Y$ and satisfies the projected Leibniz identity (\ref{defining}).


By construction, if the full connection preserves some structures, the projected connection will preserve the same structures as well. Therefore, the projection of the Born connection is a partial connection that preserves $(\eta,\omega,\HH)$. A more subtle point is the fact that if the full connection is torsionless in the generalized sense then the projected connection also is torsionless as a generalized connection. We postpone the detailed analysis of this to a future publication \cite{Freidel:new}.

Remarkably but not surprisingly, 
the Born connection can be identified with the preferred Levi-Civita connection.  More precisely we have the following theorem: 
\begin{Thm}\label{th:Born-LCconnection}
The projected Born connection $\DB$ restricted to $L=T\ell$ is the Levi-Civita connection $\n^g$ for the metric $g: L \to L^*$ 
(where $g(x) = \eta (J x)$)
\be
\DB_x y = \n^g_xy, \qquad x,y \in L.
\ee
\end{Thm}
To show this we will use the chiral projections $P_\pm$ of $x=PX\in L$: $x_\pm := \tfrac12 [x \pm g(x)]$. Then one can show that $K x_\pm = x_\mp$ (since $K$ anticommutes with $J$) and hence starting from the projected Born connection defined in  \eqref{eq:BornBracketL} we find for $y=PY$
\be
\DB_xy  = 2 P(\bl x_-,y_+ \br_{L+} + \bl x_+,y_-\br_{L -}).
\label{eq:BornonL}
\ee
To see that this is equal to the Levi-Civita connection of $g$, we need the following lemma:
\begin{Lem}
Given a metric $g$ on $L=T\ell$, viewed as a map $ g: L \to L^*$, we can express the  Levi-Civita connection of $g$ as 
\be
\n^g_x y =\tfrac12 [x,y]+ \tfrac12g^{-1} \left(\mathcal{L}_x g(y) +  \mathcal{L}_y g(x) - \rd g(x,y) \right), 
\ee
where $x,y \in L$ and $\mathcal{L}_x $ is the Lie derivative on $L^*$. This expression can equivalently be written in terms of the standard Courant bracket $[\,,\,]_\ell$ and its anchor $\pi$ as 
\be
 \n^g_x y = 4\pi([x_-,y_+]_{\ell +}) = 4\pi ([x_+,y_-]_{\ell -}).
 \label{eq:LCofg}
\ee  
\end{Lem}
The proof of this lemma is mechanical: for the first part one just checks that the expression written is a connection which is torsionless and preserves $g$; for the second part one just relates the first expression to the projection of the Courant bracket (cf. \eqref{bracketexpl} with $N=0$ since $L$ is integrable).

Now the result \eqref{eq:LCofg} can be written as $\n^g_x y = 2\pi([x_-,y_+]_{\ell +}+[x_+,y_-]_{\ell -})$ which is equal to \eqref{eq:BornonL} via the $\rho$-isomorphism \eqref{eq:D-brac-dorfman} and the correspondence $P = \pi\rho$. We thus have $\DB_x y = \n^g_xy$, completing the proof of the theorem.

\section{Conclusions and Discussion}

In the spirit of ``physics is geometry'', we have have explored structures in para-Hermitian geometry to expand the idea first formulated in \cite{Freidel:2015pka}  and further developped in \cite{Freidel:2017yuv} that spacetime is a Lagrangian subspace of a para-Hermitian manifold (see also \cite{Chatzistavrakidis:2018ztm} for a related implementation). We have introduced the notion of a D-structure which generalizes the notion of a differentiable structure -- usually encoded in the Lie bracket of a Lie algebra -- for this kind of geometry. The D-bracket (and its canonical version in terms of the canonical connection) allowed us to define generalizations of torsion and integrability. 

Born geometry, which is a para-Hermitian manifold with an additional dynamical metric, is the natural arena for the T-duality symmetric string \cite{Freidel:2014qna,Freidel:2015pka}. The D-structure together with the tools developed were then used to find a unique, torsionless connection for Born geometry (Theorem \ref{th:BornConnection}). This should be seen as a natural extension of the fundamental theorm of Riemannian geometry and the Levi-Civita connection to Born geometry.

We also discussed the relation of the D-structure to generalized geometry and the notion of a Courant algebroid (Theorem \ref{th:Dstructure-Calgebroid})  and the correspondence of the Born connection with the generalized Levi-Civita connection (Theorem \ref{th:Born-LCconnection}). It is particularly satisfactory to see that the Born connection projected onto the spacetime Lagrangian reduces to the ordinary Levi-Civita connection. This matches the fact that the structure group of Born geometry, i.e. the intersection of the three structure groups of the defining elements $(\eta,\omega,\HH)$ as in \eqref{eq:structuregroup}, is the Lorentz group $O(d)$ which is the structure group of the spacetime metric $g$. From this perspective the Born connection appears as a natural extension of the Levi-Civita connection to the doubled space and the Born geometry as the geometrical structure needed to define this extension.

An important feature of Born geometry that has not been fully explored in our work is the notion of fluxes appearing in the physics literature. We believe that the framework of Born geometry is the correct natural setting to relate the fluxes to geometrical structures and their various integrability properties. This will be discussed in the future \cite{Freidel:new} in the context of more general D-structures, i.e. ones where the D-bracket is allowed to be twisted. The key idea is to relate the D-bracket associated with non-integrable Lagrangians to generalized brackets carrying non-trivial flux. Evidence in favour of this mechanism follows from the result shown in \cite{Svoboda:2018rci}: given a para-Hermitian structure $(\PS,\eta,K)$ with D-bracket $\bracd$, we assume that there exists a reference para-K\"ahler structure $(\PS,\eta,K')$ with D-bracket $\bracd'$, then the D-bracket $\bracd$ of $K$ appears as \textit{twisted} with respect to $\bracd'$
\begin{align}
\bracd=\bracd'+\mathcal{F},
\end{align}
where $\mathcal{F}$ are the \textit{generalized fluxes} appearing in the DFT literature. From this point of view, the reference choice of integrable $K'$ serves as a choice of frame in which we express the fluxes. The interplay between the relativity of  D-structures $(\PS,\eta,K,\bracd)$ and $(\PS,\eta,K',\bracd')$  and the fluxes  will be explored further in the future.

Another limitation of our work is the fact that we have not yet included the dilaton in the construction of the bracket and the connection. Finding the proper geometrical setting that generalizes our result to the case of a non-trivial dilaton is a central challenge of future research in this direction.

Finally, given the almost para-quaternionic structure $(I,J,K)$ induced by the Born geometry, the Born connection is also very relevant to the geometry of para-quaternionic manifolds. Indeed, similar connections have been explored in this context in \cite{Ivanov:2003ze} and a special class of such connections that have non-vanishing torsion related to integrability of the almost para-quaternionic structure via Nijenhus tensors have been explored \cite{Ivanov:2004yc}. Making the relationship with the Born connection precise is an interesting question for future research.

\section*{Acknowledgements} 
F.J.R. would like to thank Chris Blair for useful discussions. L.F. would like to thank long-time collaborator D. Minic for inputs and encouragements. D.S. would like to thank his co-supervisors Ruxandra Moraru and Shengda Hu for their supervision.

The work of F.J.R. is supported by DFG grant TRR33 ``The Dark Universe''. D.S. is currently a PhD student at Perimeter institute and University of Waterloo. His research is supported by NSERC Discovery Grants 378721.

This research was supported in part by Perimeter Institute for Theoretical Physics. Research at Perimeter Institute is supported by the Government of Canada through Innovation, Science and Economic Development Canada and by the Province of Ontario through the Ministry of Research, Innovation and Science.

\appendix
\section{Further Proofs}
\label{sec:appendix}
In this appendix we provide the detailed steps for the proofs of some statements made in the main text.

\subsection{Proof of Corollary \ref{cor:BornContorsion}}
In Corollary \ref{cor:BornContorsion} the Born connection is expressed in terms of the canonical connection and the contorsion tensor. Here we want to derive this expression in detail. We start from the form \eqref{eq:BornExpanded2} which reads
\begin{equation}
\begin{aligned}
4\eta(\nB_XY,Z)&= 
	\eta\big( \bl X,Y \br - \bl J(X),J(Y) \br, Z\big) \\
	&\qquad + \eta\big( \bl X,J(Y) \br - \bl J(X),Y \br, J(Z)\big) \\
	&\qquad - \eta\big( \bl X,K(Y) \br - \bl J(X),I(Y) \br, K(Z)\big) \\
	&\qquad - \eta\big( \bl X,I(Y) \br - \bl J(X),K(Y) \br, I(Z)\big)   .
\end{aligned}
\label{eq:proofcor1}
\end{equation}
Expanding the brackets in the first two lines gives
\begin{align}
=&\eta(\nc_XY - \nc_YX - \nc_{JX}JY + \nc_{JY}JX , Z) + \eta(\nc_ZX,Y) - \eta(\nc_ZJX,JY) \notag\\
&+\eta(\nc_XJY - \nc_{JY}X - \nc_{JX}Y + \nc_YJX , JZ) + \eta(\nc_{JZ}X,JY) - \eta(\nc_{JZ}JX,Y) \\
=& 2\eta(\nc_XY,Z) - 2\eta(\nc_{JX}Y,JZ) - \eta([J\nc_{JX}J]Y-[J\nc_{JY}J]X,JZ) - \eta([J\nc_ZJ]X,Y) \notag\\
 &+ \eta([J\nc_XJ]Y+[J\nc_YJ]X,Z) - \eta([J\nc_{JZ}J]X,JY) 
\end{align}
where we expanded terms like $\eta(\nc_XJY,Z)$ as $\eta(\nc_XY,JZ)+\eta([J\nc_XJ]Y,JZ)$ using $\eta(JX,Y)=\eta(X,JY)$. We can now use
\begin{equation}
\eta([J\nc_XJ]Y,Z) = \nc_X\HH(Y,JZ)
\end{equation}
to convert the last six terms into derivatives acting on $\HH$
\begin{align}
=& 2\eta(\nc_XY,Z) - 2\eta(\nc_{JX}Y,JZ) + \nc_X\HH(Y,JZ) + \nc_Y\HH(X,JZ) \notag\\
&- \nc_{JX}\HH(Y,Z) + \nc_{JY}\HH(X,Z) - \nc_Z\HH(X,JY) - \nc_{JZ}\HH(X,Y) .
\end{align}
The third and fourth line in \eqref{eq:proofcor1} can be obtained from the first two by the replacement $(X,Y,Z)\rightarrow(X,KY,KZ)$ and an overall minus sign. These two lines then read (recall $I=JK$)
\begin{align}
=&  -\big[ 2\eta(\nc_XKY,KZ) - 2\eta(\nc_{JX}KY,IZ) + \nc_X\HH(KY,IZ) + \nc_{KY}\HH(X,IZ) \notag\\
&\quad- \nc_{JX}\HH(KY,KZ) + \nc_{IY}\HH(X,KZ) - \nc_{KZ}\HH(X,IY) - \nc_{IZ}\HH(X,KY) \big]  .
\end{align}
In a final step we use $\eta(KX,KY)=-\eta(X,Y)$, $\nc K = K\nc$ and
\begin{align}
\nc_X \HH(KY,Z)-\nc_X \HH(Y,KZ)&= 0 \\
\nc_X \HH(JY,Z)+\nc_X \HH(Y,JZ)&= 0 \\
\nc_X\Hh(IY,Z)-\nc_X \HH(Y,IZ)&= 0
\end{align}
to get some more cancellations and combine everything into
\begin{align}
4\eta(\nB_XY,Z)&= 4\eta(\nc_XY,Z) + 2\nc_X \HH(Y,JZ) \\
&\qquad + \Big[\nc_Y\Hh(JZ,X) + \nc_{JY}\Hh(Z,X) - \nc_{KY}\Hh(IZ,X) - \nc_{IY}\Hh(KZ,X)\Big] \notag\\
&\qquad - \Big[\nc_Z\Hh(JY,X)  + \nc_{JZ}\Hh(Y,X) -\nc_{KZ}\Hh(IY,X) - \nc_{IZ}\Hh(KY,X) \Big] \notag
\end{align}
which is the desired result for the contorsion of the Born connection given in Corollary \ref{cor:BornContorsion},  \eqref{eq:BornContorsion}.

\subsection{Proof of Lemma \ref{lem:LemmaForTorsion}}
In Section \ref{sec:proof} the generalized torsion of the Born connection is expressed in terms of the generalized Nijenhuis tensor which subsequently vanishes. This computation uses a lemma which states
\begin{align}
\eta(\bl X_\mp,Y_\pm\br,Z_\pm) &= - \eta(\bl Y_\pm,X_\mp\br,Z_\pm) \\
\eta(\bl X_\pm,Y_\pm\br,Z_\mp) &= \eta(\bl Z_\mp,X_\pm\br,Y_\pm) .
\end{align}
and which we will prove here. Expanding the LHS of the first expression with the upper sign reads
\begin{equation}
\begin{aligned}
\eta(\bl X_-,Y_+\br,Z_+) &= \eta(\nc_{X_-}Y_+ - \nc_{Y_+}X_-,Z_+) + \eta(\nc_{Z_+}X_-,Y_+) \\
	&= -[\eta(\nc_{Y_+}X_- - \nc_{X_-}Y_+,Z_+) + \eta(\nc_{Z_+}Y_+,X_-)] \\
	&\qquad + \eta(\nc_{Z_+}Y_+,X_-) + \eta(\nc_{Z_+}X_-,Y_+) \\
	&= -\eta(\bl Y_+,X_-\br,Z_+) + Z_+[\eta(X_-,Y_+)] - (\nc_{Z_+}\eta)(X_-,Y_+) \\
	&= -\eta(\bl Y_+,X_-\br,Z_+).
\end{aligned}
\end{equation}
In the final step the last two terms vanish since $\eta(X_-,Y_+)=0$ as in \eqref{eq:Cpm-orthogonal-property} and $\nc\eta=0$ by the metric compatibility of the canonical connection. The same of course holds for the opposite sign choice which thus proves the first part of the lemma.

Similarly expanding the LHS of the second expression with the upper sign reads
\begin{equation}
\begin{aligned}
\eta(\bl X_+,Y_+\br,Z_-) &= \eta(\nc_{X_+}Y_+ - \nc_{Y_+}X_+,Z_-) + \eta(\nc_{Z_-}X_+,Y_+) \\
	&= \eta(\nc_{X_+}Y_+ - \nc_{Y_+}X_+,Z_-) \\
	&\qquad + [\eta(\nc_{Z_-}X_+ - \nc_{X_+}Z_-,Y_+) + \eta(\nc_{Y_+}Z_-,X_+)] \\
	&\qquad + \eta(\nc_{X_+}Z_-,Y_+) - \eta(\nc_{Y_+}Z_-,X_+) \\
	&= \eta(\bl Z_-,X_+\br,Y_+) + X_+[\eta(Y_+,Z_-)] - \nc_{X_+}\eta(Y_+,Z_-) \\
	&\qquad - Y_+[\eta(X_+,Z_-)] + \nc_{Y_+}\eta(X_+,Z_-) \\
	&= \eta(\bl Z_-,X_+\br,Y_+)
\end{aligned}
\end{equation}
where in the final step terms vanish for the same reasons as above. Again this also holds for the opposite sign choice thus proving the second part of the lemma.

\subsection{Proof of the Jacobi Identity}
In this section we prove the Jacobi identity for the $L$-bracket
\be
\bl x+\tilde{x} ,y+\tilde{y} \br_L =
[x,y] + \Lie_x \yt - \Lie_y \xt + \Db \eta(x,\yt)
\ee
in the case $L$ is integrable.
We define the Jacobiator to be 
\be 
J(X,Y,Z):= \bl \bl X,Y\br_L, Z\br_L-
\bl X, \bl Y, Z\br_L\br_L 
 + 
 \bl Y, \bl X, Z\br_L\br_L.
\ee
One first looks at the symmetry property of the bracket and establishes that 
\be
\bl \Db f , y+ \yt \br_L =0,\qquad 
\bl x+\xt , \Db f  \br_L = \Db x[f].
\ee
which follows from the evaluations
\begin{equation}
\begin{aligned}
\eta(\bl \Db f , y+ \yt \br_L, z) &= - \eta(L_y \Db f,z) + z[y[f]] \\
  &= - y[ z[f]] + [y,z][f] + z[y[f]] =0, \\
\eta(\bl x+\xt , \Db f  \br_L, z) &=  \eta(\Lie_x \Db f, z) \\
&= x[z[f]]- \eta(\Db f , [x,z]) = z[x[f]].
\end{aligned}
\end{equation}
Now we look at symmetric components of the Jacobiator
\begin{equation}
\begin{aligned}
J(X,X,Z)&:= \bl \bl X,X\br_L, Z\br_L = \bl \Db\eta(x,\xt),Z\br_L=0,\\
J(X,Y,Y)&:= \bl \bl X,Y\br_L, Y\br_L-\bl X, \bl Y, Y\br_L\br_L +  \bl Y, \bl X, Y\br_L\br_L \\
 &= \Db\eta(\bl X, Y\br_L, Y) - \tfrac12 \Db x[\eta(Y,Y)] =0.
\end{aligned}
\end{equation}
The polarised Jacobiator component $J(x,y,z)$ is simply given by
\bea
J(x,y,z) =[ [x,y], z] -[x,[y,z]] +[y,[x,z]]=0.
\eea
Next we show that the Lie derivative is a morphism
\begin{equation}
\begin{aligned}
\eta( \Lie_x  \Lie_y \zt , w) &= x[\eta(   \Lie_y \zt , w)]- \eta(   \Lie_y \zt , [x,w])\\
 &= x[y[\eta( \zt , w)]- \eta(\zt, [y,w])]- \eta(   \Lie_y \zt , [x,z])\\
 &= x[y[\eta( \zt , w)]]- x[\eta(\zt, [y,w])]- y\eta(   \Lie_y \zt , [x,z]) +\eta(  \eta\zt, [y,[x,z]).
\end{aligned}
\end{equation}
Taking the skew combination therefore gives
\begin{equation}
\begin{aligned}
\eta( [\Lie_x, \Lie_y] \zt , w) &= [x,y][\eta( \zt , w)]+ \eta( \eta\zt, [y,[x,z]-[x,[y,z])\\
 &= [x,y][\eta( \zt , w)]- \eta( \eta\zt,[[x,y,z]]) = \Lie_{[x,y]} \zt .
\end{aligned}
\end{equation}
Finally, the Jacobiator with two elements in $\Lt$ vanishes since $\Lt$ is abelian. For instance
\begin{equation}
\begin{aligned}
 J(x,\yt,\zt) &:= \bl \bl x,\yt \br_L, \zt \br_L-\bl x, \bl \yt , \zt \br_L\br_L + \bl \yt , \bl x, \zt \br_L\br_L \\
 &= \bl \Lie_x\yt, \zt \br_L + \bl \yt , \Lie_x \zt \br_L =0.
\end{aligned}
\end{equation}
This completes the proof.

\bibliographystyle{JHEP} 
\bibliography{mybib}

\end{document}